\documentclass[10pt]{article}
\usepackage{enumitem}
 \usepackage{amsfonts}
\usepackage{amsmath}
\usepackage{graphicx}
\usepackage{subfigure}
\usepackage{fancyhdr}
\usepackage{amsthm}
\usepackage{amssymb}
\usepackage{amscd}
\usepackage{float}
\newtheorem{definition}{Definition}[section]

\newtheorem{lemma}{Lemma}[section]
\newtheorem{proposition}{Proposition}[section]
\newtheorem{remark}{Remark}[section]

\begin{document}
\title{{Variance swaps under L\'{e}vy process with stochastic volatility and stochastic interest rate in incomplete markets}\thanks{This work was supported by  the National Natural Science Foundation of China (11471230, 11671282).}}
\author{{Ben-zhang Yang$^a$, Jia Yue$^b$ and Nan-jing Huang$^a$\thanks{Corresponding author.  E-mail address: nanjinghuang@hotmail.com}}\\
{\small\it a. Department of Mathematics, Sichuan University, Chengdu, Sichuan 610064, P.R. China}\\
{\small\it b. Department of Economic Mathematics, South Western University of Finance and Economics,}\\
{\small\it Chengdu, Sichuan 610074, P.R. China}}
\date{}
\maketitle
\vspace*{-9mm}
\begin{center}
\begin{minipage}{5.5in}
{\bf Abstract.}
This paper focuses on the pricing of the variance swap in an incomplete market where the stochastic interest rate and the price of the stock are respectively driven by Cox-Ingersoll-Ross model and Heston model with simultaneous L\'{e}vy jumps. By using the equilibrium framework, we obtain the pricing kernel and the equivalent martingale measure. Moreover, under the forward measure instead of the risk neural measure, we give the closed-form solution for the fair delivery price of the discretely sampled variance swap by employing the joint moment generating function of the underlying processes.  Finally,  we provide some numerical examples to depict that the values of variance swaps not only depend on the stochastic interest rates but also increase in the presence of jump risks.

\end{minipage}
\end{center}
\section{Introduction}

In terms of the financial markets, volatility has always been considered as a key measure. The financial development and growth over last century has caused the role of volatility to change. Volatility derivatives generally are special financial tools which provide opportunities to display the financial market fluctuations and give methods to manage  volatility risks for investors (see, e.g., \cite{Bates,Bernard,Cao,Carr,Mon,Ruan,Windcliff,Zheng,Zhu}).  They are traded for decision-making between long or short positions, trading spreads between realized and implied volatility, and hedging against volatility risks. Of all volatility derivatives, variance swaps are written on underlying assets' historical volatility and they are related to previous standard deviation of financial returns involving a specified time period. Various theoretical results, numerical algorithms and applications have been studied extensively for variance swaps in the literature; for instance we refer the reader to \cite{Carr,Coqueret,Pun,Ruan,Shen,Zheng} and the references therein.

Along with the development of studies concerned with variance swaps, not only do investors in financial markets care for correct prices of variance swaps, but researchers in financial mathematics also attempt to construct practical models and provide feasible methods for pricing variance swaps under more weak assumptions. There are many researchers to focus on pricing variance swaps (or other volatility derivatives) by developing suitable market models for estimating values of variance swaps, in which three vital market factors are considered: the stochastic volatility, the jump diffusion and the stochastic interest rate. For example, the stochastic volatility of underlying assets is widely applied to avoid volatility smiles in financial markets, such as the Heston model (see, e.g., \cite{Fonseca,Heston}). The jump diffusion is extensively applied to describe non-Gaussian characters of assets' returns, such as the Geometric Variance Gamma model, the Merton model and the Geometric Stable Process Model (see, e.g., \cite{Huang,Madan}). The stochastic interest rate is often used to model the uncertainty of interest rate in the financial markets, such as the Cox-Ingersoll-Ross (for short CIR) model, the Hull-White (for short HW) model and the Heath-Jarrow-Morton (for short HJM) model (see, e.g., \cite{Brigo,Falini,Grzelak,Heath}).

Recently, based on the Heston stochastic volatility model, Grunbichler et al. \cite{Grunbichler} built a new pricing model for options on variance and showed the fundamental difference between the volatility derivative and the equity option.  By using the static replication of options, Carr and Madan \cite{Carr} studied the pricing and hedging of variance swaps without specifying the volatility process. Heston and Nandi \cite{Heston2} noted that specifying the mean reverting square root process has the disadvantage of unobservable underlying assets and then proposed a new model with the advantage hedging various volatility derivatives by using only a single asset. To distinguish the states of a business cycle, Elliott et al. \cite{Elliott1} constructed a continuous-time Markovian-regulated version of the Heston stochastic volatility model to price variance swaps, delivered analytical formulas by adopting the regime-switching Esscher transform, and showed that the valuation of volatility swaps based on the stochastic volatility model was significantly higher than one without the stochastic volatility model.

Incorporating the jump diffusion into models of pricing and hedging variance swaps, Carr et al. \cite{Carr2} and Huang et al. \cite{Huang} studied the pricing of variance swaps with time-changed L\'evy processes and found the fact that many small jumps cannot be adequately modelled by using finite-activity compound Poisson processes. Recently, Zheng and Kwok \cite{Zheng} presented a general analytic approach for pricing discretely sampled variance swaps under the stochastic volatility models with simultaneous jumps. Through analytic calculations, they proved that pricing formulas for the discretely sampled generalized variance swaps converge to ones for the continuously sampled variance swaps. Very recently, Cui et al. \cite {Cui} proposed a general framework for discretely sampled realized variance derivatives in stochastic volatility models with jumps. As pointed out by Cui et al. \cite {Cui},  the framework proposed in \cite {Cui} encompassed and extended some previous models  on discretely sampled volatility derivatives and provided highly efficient valuation methods.

On the other hand, taking into account the stochastic interest rate in financial markets, Kim et al. \cite{Kim} proposed a model by combining the multi-scale stochastic volatility model with the Hull-White interest rate model and showed that the values of variance swaps depend heavily on the variety of interest rates. Recently, by considering the effects of stochastic interest rates and the stochastic regime-switching volatility for pricing variance swaps, Shen et al. \cite{Shen} investigated the pricing model involving the stochastic interest rate and the stochastic regime-switching volatility. They demonstrated the effect of both the stochastic interest rate and the regime-switching is  significant in pricing variance swaps by providing a numerical analysis for the case with a two-state Markov chain. Very recently, Cao et al. \cite{Cao} and Roslan et al. \cite{Roslan} studied the effects of imposing stochastic interest rates driven by the CIR process along with the Heston stochastic volatility model for pricing variance swaps under the discrete sampling times. Some numerical results were given \cite{Cao, Roslan} to support the fact that the interest rates could impact and change the values of variance swaps.

Taken all together, we note that the values of variance swaps are dependent on the stochastic volatility, the jump diffusion and the stochastic interest rate. Thus, it would be important and interesting to consider the stated three factors in pricing variance swaps.  However, to the best of our knowledge, the work has not been reported in the literature for pricing variance swaps based on the stochastic volatility with simultaneous jumps and the stochastic interest rate. The main purpose of this paper is to make an attempt in this direction.  We construct a hybridization model for pricing variance swaps in financial markets, in which the stochastic interest rate is driven by the Cox-Ingersoll-Ross model and the volatility of the stock is described by the Heston model with simultaneous L\'{e}vy jumps. We obtain the closed-form solution for the fair delivery price of the discretely sampled variance swap via the analytical expression of the joint moment generating function of the underlying processes. We also give some numerical experiments to support the main results of this paper.

The rest of this paper is organized as follows. Section 2 gives some analytical formulas for pricing of variance swaps under the stochastic volatility model with simultaneous jumps and the stochastic interest rate in partial correlation case. In Section 3, we derive analytical formulas for pricing of variance swaps by solving the model given in full correlation case.  Some numerical examples for pricing variance swaps are reported in Section 4. Finally, conclusions are given in Section 5.

\section{Stochastic volatility model with jumps and CIR model: partial correlation case}

Let $(\Omega,\mathcal{F},\{\mathcal{F}_t\}_{t\geq 0},\mathbb{P})$ be an underlying filtered complete probability space with a physical probability measure $\mathbb{P}$. In this paper, we assume that the stochastic market interest rate is driven by the CIR model and the asset price is formulated by the stochastic volatility model with simultaneous jumps, respectively. More precisely, the price of the underlying asset $S(t)$, its instantaneous volatility $V(t)$ and the market interest rate $r(t)$ are governed by the following system of stochastic differential equations:
\begin{eqnarray}\label{real}
\left\{
\begin{array}{ll}
dS(t)=\mu S(t-)dt+\sqrt{V(t)}S(t-)dW_{1}(t)+\int_{R} (e^x-1)S(t-)\widetilde{N}_x(dt,dx),\quad S(0)=S_0>0,\\
dV(t)=\kappa(\theta-V(t))dt+\sigma\sqrt{V(t)}(\rho dW_{1}(t)+\sqrt{1-\rho^2} dW_{2}(t)),\quad V(0)=V_0>0,\\
dr(t)=\alpha(\beta-r(t))dt+\eta \sqrt{r(t)}dW_{3}(t),\quad r(0)=r_0>0,
\end{array}
\right.
\end{eqnarray}
where $R^d=\mathbb{R}^d\backslash\{\mathbf{0}\}$ for  $d>1$ and $R^1=(-\infty, +\infty)$;  $S(t-)$ stands for the value of $S(t)$ before a possible jump occurs; $W_i(t)$ $(i=1,2,3)$ are three independent Brownian motions; $\widetilde{N}_x(dt,dx)=N_x(dt,dx)-\nu_x(dx)dt$ is a compensated jump measure with respect to the jump measure $N_x(dt,dx)$ and the L\'evy kernel (density) $\nu_x(dx)$ satisfying $
\int_R \min\{1,x^2\}\nu_x(dx)<\infty$; the drift term $\mu>0$ is the expected return of the stock; the parameters $\kappa$, $\theta$ and $\sigma$ are the mean-reverting speed, the long-term mean and  the volatility of volatility (vol of vol) in the instantaneous volatility process $V(t)$, respectively; the parameters $\alpha$, $\beta$ and $\eta$  determine the speed of mean reversion, the long-term mean and the volatility of the interest rate in the stochastic instantaneous interest rate $r(t)$, respectively; $-1 \leq \rho \leq 1$ is the correlation coefficient between the stock price and the volatility.  Furthermore, in order to ensure that the square root processes are always positive, it is required that $2\kappa\theta\geq \sigma^2$ and $2\alpha\beta\geq\eta^2$ (see \cite{Heston}).    %$x\in R-\{0\}$.

Some special cases of \eqref{real} are as follows.
\begin{itemize}
\item[(i)] If interest rate process $r(t)$ is a constant, then system \eqref{real} reduces to the model considered by Ruan et al. \cite{Ruan}, Cui et al. \cite{Cui} and Zheng et al. \cite{Zheng}.
\item[(ii)] If there the jump diffusion is removed, then system \eqref{real} reduces to the model considered by Cao et al. \cite{Cao}, Roslan et al. \cite{Roslan} and Shen et al. \cite{Shen}.
\item[(iii)] If interest rate process $r(t)$ is a constant and there is no jump diffusion, then system \eqref{real} reduces to the model considered by Carr et al. \cite{Carr,Carr2} and Zhu et al. \cite{Zhu}.
\end{itemize}
\subsection{Pricing kernel}
We note that there are four uncertainties in the system driven by three Brownian motions and one jump process. It is well known that the pricing kernel is critical for determining the pricing and hedging of assets. In general, for any asset $p$ and its cash flow at time $T$, the pricing kernel $\pi$ should satisfy
$$p(t)=\mathbb{E}_t^{\mathbb{P}}\left[\frac{\pi(T)}{\pi(t)}p(T)\right],$$
where $\mathbb{E}_t^{\mathbb{P}}[\cdot]$ is the conditional expectation at time $t$ in the physical probability measure $\mathbb{P}$. It is also called the martingale condition which requires that the multiply $\pi(t)p(t)$ is a martingale. Thus, the pricing kernel should follow the restriction
$$p(t)\pi(t)=\mathbb{E}_t^{\mathbb{P}}\left[\pi(T)p(T)\right].$$

Due to the existence of the jump component, the market considered in this paper is incomplete and so there are infinitely many equivalent martingale measures for the asset pricing. Here we shall employ the idea of equilibrium pricing method to find an equivalent martingale measure and determine the corresponding risk-neural price processes of risk assets.  To this end,  we consider a money market account whose price M(t) with interest rate $r(t)$ as follows
$$\frac{dM(t)}{M(t)}=r(t)dt.$$

If we define the expected excess return of the stock (equity premium)
as $\phi :=\mu-r(t)$, then the expected return of the stock will be decomposed
into two parts: the interest rate and the equity premium. We assume that there is a representative investor with the portfolio $(u(t), 1-u(t))$ which means the fraction of wealth invested in the stock and the money market, respectively. The consumption rate of the investor is assumed to be $c(t)$. Then the investor's wealth process $\widetilde{W}(t)$ satisfies the following stochastic differential equations:
\begin{equation}
\left\{
\begin{aligned}
&\frac{d\widetilde{W}(t)}{\widetilde{W}(t)}=\left(r(t)+\phi u(t)-\frac{c(t)}{\widetilde{W}(t)}\right)dt+\sqrt{V(t)}dW_{1}(t)+\int_{R} (e^x-1)\widetilde{N}_x(dt,dx),\\
&dV(t)=\kappa(\theta-V(t))dt+\sigma\sqrt{V(t)}(\rho dW_{1}(t)+\sqrt{1-\rho^2} dW_{2}(t)),\\
& dr(t)=\alpha(\beta-r(t))dt+\eta \sqrt{r(t)}dW_{3}(t).
\end{aligned}
\right.
\end{equation}

Moreover, we assume that the representative investor has a CRRA utility
\begin{equation}\label{CRRA}
U(c)=\frac{c(t)^{1-\vartheta}}{1-\vartheta},
\end{equation}
where the relative risk aversion coefficient $\vartheta>0$ and $\vartheta\neq 1$. Choosing the portfolio $u(t)$ and the consumption rate $c(t)$,
the representative investor maximizes naturally his/her expected objective
function \eqref{CRRA} in an infinite horizon, that is,
\begin{align}\label{max}
&\max_{u,c}\quad E\left[\int_0^\infty e^{-\delta(s-t)}\frac{{c(s)}^{1-\vartheta}}{1-\vartheta}ds\right],\\
&s.t.\quad \left\{
\begin{aligned}\label{C}
&\frac{d\widetilde{W}(t)}{\widetilde{W}(t)}=\left(r(t)+\phi u(t)-\frac{c(t)}{\widetilde{W}(t)}\right)dt+\sqrt{V(t)}dW_{1}(t)+\int_{R} (e^x-1)\widetilde{N}_x(dt,dx),\\
&dV(t)=\kappa(\theta-V(t))dt+\sigma\sqrt{V(t)}(\rho dW_{1}(t)+\sqrt{1-\rho^2} dW_{2}(t)),\\
&dr(t)=\alpha(\beta-r(t))dt+\eta \sqrt{r(t)}dW_{3}(t),
\end{aligned}
\right.
\end{align}
where the time discount parameter $\delta>0$ is a constant. Based on the studies of \cite{Filipovica,Fu,Ruan}, the market equilibrium can be defined
in a standard way as follows.

\begin{definition}\label{equl} The market equilibrium occurs when the representative investor maximizes his/her expected objective function \eqref{max} and the market is cleared, that is, $u(t) = 1$.
\end{definition}

After solving the optimal portfolio-consumption problem \eqref{max} by
using the Hamilton-Jacobi-Bellman (HJB) method under the equilibrium condition, we get the following proposition.

\begin{proposition}\label{premium} In the production (the stock market) economy
with a representative investor who has CRRA utility and with a
production process governed by \eqref{real}, the
equilibrium equity premium $\phi$ is given by
\begin{equation}\label{kerneltest}
\phi=(\vartheta-\sigma \rho I)V(t)+\int_R(e^x-1)(1-e^{-\vartheta x})\nu_x(dx),
\end{equation}
where $I$, $K$ and $M$ are determined by the following equations
\begin{equation}\label{HJBO}
\left\{
\begin{aligned}
&\Gamma+\kappa\theta I+\alpha \beta K+\vartheta e^{-\frac{M+I \theta+K\beta}{\vartheta}}(1+\frac{\theta I}{\vartheta}+\frac{\beta K}{\vartheta})=0,\\
&\frac{1}{2}\vartheta(1-\vartheta)-\kappa I+\frac{1}{2}\sigma^2 I^2-Ie^{-\frac{M+I \theta+K\beta}{\vartheta}}=0,\\
&1-\vartheta-\alpha K+\frac{1}{2}\eta^2 K^2-Ke^{-\frac{M+I \theta+K\beta}{\vartheta}}=0
\end{aligned}
\right.
\end{equation}
with
$$\Gamma=-\delta-(1-\vartheta)\int_R\left(e^{(1-\vartheta)x}-e^{\vartheta x}\right)\nu_x(dx)+\int_R\left(e^{(1-\vartheta)x}-1\right)\nu_x(dx).$$
\end{proposition}
\begin{proof}
Let
$$J(\widetilde{W},V,r)=\max_{u,c} E\left[\int_0^\infty e^{-\delta(s-t)}\frac{{c(s)}^{1-\vartheta}}{1-\vartheta}ds\right].$$
Then $J$ satisfies the following HJB equation
\begin{equation}\label{HJB}
\max_{u,c}\left\{-\delta J+\mathcal{L}J+\frac{c^{1-\vartheta}}{1-\vartheta}\right\}=0,
\end{equation}
where
\begin{equation*}
\begin{aligned}
\mathcal{L}J=&J_{\widetilde{W}}
\widetilde{W}\left(r+\phi u-\frac{c}{\widetilde{W}}\right)
+J_V\kappa(\theta-V)+\frac{1}{2}J_{VV}\sigma^2V\\
&+J_r\alpha(\beta-r)+\frac{1}{2}J_{rr}\eta^2r
+\frac{1}{2}J_{\widetilde{W}\widetilde{W}}\widetilde{W}^2Vu^2
+J_{\widetilde{W}V}\widetilde{W}V\rho\sigma u\\
&+\int_R\left[J(\widetilde{W}(1+(e^x-1)u,V,r)-J(\widetilde{W},V,r)-J_{ \widetilde{W}}\widetilde{W}(e^x-1)u)\right]\nu_x(dx).
\end{aligned}
\end{equation*}
This leads to the first-order condition for optimal problem \eqref{max} with constraints \eqref{C} as follows:
\begin{eqnarray}\label{condition}
\left\{
\begin{aligned}
&J_{\widetilde{W}}
\widetilde{W}\left(\phi+(e^x-1)u)\nu_x(dx)\right)+J_{\widetilde{W}\widetilde{W}}\widetilde{W}^2Vu
+J_{\widetilde{W}V}\widetilde{W}V\rho\sigma\\
&\quad \mbox{}+\int_R\left[J_{\widetilde{W}}(\widetilde{W}(1+(e^x-1)u,V,r)(e^x-1)u)\widetilde{W}\right]\nu_x(dx)=0,\\
&-J_{\widetilde{W}}+c^{-\vartheta}=0.
\end{aligned}
\right.
\end{eqnarray}
Following the equilibrium condition in Definition \ref{equl} and taking $u=1$, we have
\begin{equation}\label{kernelformula}
\begin{aligned}
\phi= -\frac{1}{J_{\widetilde{W}}}\bigg(J_{\widetilde{W}\widetilde{W}}\widetilde{W}V
+J_{\widetilde{W}V}\widetilde{W}V\rho\sigma
+\int_RJ_{\widetilde{W}}(\widetilde{W}e^x,V,r)(e^x-1))\nu_x(dx)\bigg)
+\int_R(e^x-1)\nu_x(dx).
\end{aligned}
\end{equation}
Substituting \eqref{condition} and \eqref{kernelformula} into \eqref{HJB}, we get the following partial differential condition
\begin{equation}\label{HJBP}
\begin{aligned}
0=&J_{\widetilde{W}}
\widetilde{W}r\widetilde{W}
+J_V\kappa(\theta-V)+\frac{1}{2}J_{VV}\sigma^2V
+J_r\alpha(\beta-r)+\frac{1}{2}J_{rr}\eta^2r\\
&-\frac{1}{2}J_{\widetilde{W}\widetilde{W}}\widetilde{W}^2V
+\int_R\left[J(\widetilde{W}e^x,V,r)-J(\widetilde{W},V,r)\right]\nu_x(dx)\\
&-\widetilde{W}\int_RJ_{ \widetilde{W}}(\widetilde{W}e^x,V,r)(e^x-1)\nu_x(dx)+\frac{\vartheta}{1-\vartheta}J_{\widetilde{W}}^{1-\frac{1}{\vartheta}}-\vartheta J.
\end{aligned}
\end{equation}
Suppose that the value function has the following form:
\begin{equation}\label{Jform}
J(\widetilde{W},V,r)=e^{M+I V+Kr}\frac{{\widetilde{W}}^{1-\vartheta}}{1-\vartheta}.
\end{equation}
Then, by substituting \eqref{Jform} into \eqref{HJBP},  one has
\begin{equation}\label{HJBO}
\left\{
\begin{aligned}
&\Gamma+\kappa\theta I+\alpha \beta K+\vartheta e^{-\frac{M+I \theta+K\beta}{\vartheta}}(1+\frac{\theta I}{\vartheta}+\frac{\beta K}{\vartheta})=0,\\
&\frac{1}{2}\vartheta(1-\vartheta)-\kappa I+\frac{1}{2}\sigma^2 I^2-Ie^{-\frac{M+I \theta+K\beta}{\vartheta}}=0,\\
&1-\vartheta-\alpha K+\frac{1}{2}\eta^2 K^2-Ke^{-\frac{M+I \theta+K\beta}{\vartheta}}=0
\end{aligned}
\right.
\end{equation}
with
$$\Gamma=-\delta-(1-\vartheta)\int_R\left(e^{(1-\vartheta)x}-e^{\vartheta x}\right)\nu_x(dx)+\int_R\left(e^{(1-\vartheta)x}-1\right)\nu_x(dx).$$
Combining \eqref{Jform} and \eqref{HJBO} with \eqref{kernelformula}, the results follow immediately.
\end{proof}

From Proposition \ref{premium}, we know that risk premiums are contributed by three Brownian diffusion risks and one jump diffusion risk with small jump sizes. Therefore, the pricing kernel is related to four risk sources and so we can define the pricing kernel $\pi(t)$ as follows:
\begin{equation}\label{pi}
\frac{d \pi(t)}{\pi(t)}=-r(t)dt-\gamma_1(t)dW_1(t)-\gamma_2(t)dW_2(t)-\gamma_3(t)dW_3(t)+\int_R (e^z-1)\widetilde{N}_z(dt,dz),
\end{equation}
where $\gamma_i(t)$ is the market price of the $i$-th diffusion risk (risk premium) from $W_i(t)$  and the new compensated jump measure of $\pi(t)$ is formulated by
\begin{eqnarray}\label{e+1}
\widetilde{N}_z(dt,dz)={N}_z(dt,dz)-\nu_z(dz)dt
\end{eqnarray}
here $N_z(dt,dz)$ is a new jump measure and $\nu_z(dz)$ is a new L\'{e}vy kernel for $z$. Based on the analysis of Proposition \ref{premium}, we know that the risk premium contributed by the Brownian motion $W_1(t)$ is $(\vartheta-\sigma \rho I)V(t)$. Then the market price of the first diffusion risk from $W_1(t)$ is given by
$$\gamma_1(t)=\frac{(\vartheta-\sigma \rho I)V(t)}{\sqrt{V(t)}}=(\vartheta-\sigma \rho I)\sqrt{V(t)}.$$
Moreover, making use of the martingale condition of pricing kernel $\pi(t)$, we can get the pricing kernel under the production economy by the following proposition.

\begin{proposition}\label{lemma1}
In the production economy, the pricing kernel $\pi(t)$ for (\ref{real}) satisfies the following equation
\begin{equation}\label{pi2}
\frac{d \pi(t)}{\pi(t)}=-r(t)dt-(\vartheta-\sigma \rho I)\sqrt{V(t)}dW_1(t)-\gamma_2(t)dW_2(t)-\gamma_3(t)dW_3(t)+\int_R (e^z-1)\widetilde{N}_z(dt,dz),
\end{equation}
where $\widetilde{N}_z(dt,dz)$ is given by \eqref{e+1} such that
\begin{eqnarray}\label{cons}
\int_{R}(e^x-1)e^{-\vartheta x}\nu_{x}(dx) =
\int_{R^2} e^z(e^x-1)\nu_{z,x}(dz\times dx).
\end{eqnarray}
%\begin{eqnarray}\label{cons}
%0 &=&-
%[\sqrt{V(t)}-\frac{\mu-r(s)}{\sqrt{V(s)}}]^2
%-\gamma^2_{2}(t)-\gamma_3^{2}(t)\nonumber\\
%&&+\int_{R_0} [z-(e^z-1)]\nu_z(dz)
%+\int_{R_0} [x-(e^x-1)]\nu_x(dx)\nonumber\\
%&&+\int_{R_0}[\ln(1+(e^z-1)(e^x-1))-(e^z-1)(e^x-1)]\nu_{z,x}(dz\times dx)\nonumber\\
%&&+\int_{R^2} (e^z-1)(e^x-1)\nu_{z,x}(dz\times dx)
%\end{eqnarray}
\end{proposition}
\begin{proof}
The proof is similar to the one of Proposition 5 in \cite{Fu} and so we omit it here.
\end{proof}

\begin{remark}
We note that the risk premium $\gamma_2(t)$ with respect to the volatility $V(t)$ can not be obtained by applying the martingale condition to Brownian motion $W_2(t)$. In fact, the risk premium $\gamma_2(t)$ depends not only on the volatility $V(t)$ but also on the price of stock $S(t)$. The same holds true for the risk premium $\gamma_3(t)$ with respect to the interest rate $r(t)$. However, the distribution of the L\'evy kernel $v_{x,z}$ in (\ref{pi2}) can be arbitrary, as long as it satisfies (\ref{cons}) (see, for example, \cite{Carr,Carr2,Madan}).
\end{remark}

Generally, in order to price derivatives, we attempt to look for a so-called equivalent martingale measure $\mathbb{Q}$ such that $\mathbb{Q}\sim \mathbb{P}$. It is not easy to determine $\gamma_2(t)$, $\gamma_3(t)$ and the distribution of jump process in the physical measure $\mathbb{P}$, so we need to get the particular forms of the stochastic processes in a risk-neutral measure $\mathbb{Q}$. Define a Radon-Nikod\'{y}m derivative as follows
\begin{equation}\label{rn}
Z(t):=\frac{d\mathbb{Q}}{d\mathbb{P}}:=e^{\int_0^t r(s)ds}\pi(t).
\end{equation}
Then the asset pricing formula can be rewritten by following lemma under the risk-neutral measure $\mathbb{Q}$.
\begin{lemma}\label{risk neutral pricing}
Under the risk-neutral measure $\mathbb{Q}$, the asset pricing formula is given by
$$p(t)=\mathbb{E}_t^{\mathbb{Q}}\left[e^{-\int_t^T r(s)ds }p(T)\right].$$
\begin{proof}
By the definition of pricing kernel, we have
%\begin{equation*}
$$e^{-\int_0^t r(s)ds}p(t)=e^{-\int_0^t r(s)ds} \mathbb{E}_t^{\mathbb{P}}\left[\frac{\pi(T)}{\pi(t)}p(T)\right]=\mathbb{E}_t^{\mathbb{P}}\left[e^{-\int_0^T r(s)ds}\frac{Z(T)}{Z(t)}p(T)\right]=\mathbb{E}_t^{\mathbb{Q}}\left[e^{-\int_0^T r(s)ds}p(T)\right],
$$%\end{equation*}
which just proves our result since $e^{-\int_0^t r(s)ds}$ is $\mathcal{F}_t$-measurable.
\end{proof}
\end{lemma}

From Proposition \ref{lemma1}, Lemma \ref{risk neutral pricing} and Girsanov Theorem (see \cite{Mon,BO,Ruan,Shen}),  we know that the stochastic system (\ref{real})  can be changed from $\mathbb{P}$ into $\mathbb{Q}$. Thus, Brownian motions, the compensated jump measure and the parameters in (\ref{real}) can be transformed from $\mathbb{P}$ into $\mathbb{Q}$. In order to ensure the mean-reverting speed of the volatility and the interest rate in (\ref{real}) under $\mathbb{Q}$ are both constants, we employ Heston's assumptions on the parameters of mean-reverting process.

As illustrated in \cite{Heston}, Heston applied Breeden's consumption-based model to yield a volatility risk premium having form $\lambda(t,S(t),\nu(t))=\lambda \nu(t)$ for the CIR square-root process $\nu(t)$. Inspired by Heston's idea, in this paper, we assume that the risk premiums of volatility $V(t)$ and interest rate $r(t)$ as follows
$$\lambda_1(t,S(t),V(t),r(t))=\lambda_1 V(t), \quad \lambda_2(t,S(t),V(t),r(t))=\lambda_2 r(t),$$
where $\lambda_1$ and $\lambda_2$ are two constants. Then it follows from \cite{Heston} that parameters in the pricing kernel $\pi$ can be given by
$\gamma_2(t)=\lambda_1\sqrt{V(t)}/\sigma$ and $\gamma_3(t)=\lambda_2\sqrt{r(t)}/\eta$, respectively. To obtain the new form of \eqref{real} under $\mathbb{Q}$, we consider transformations given by
\begin{eqnarray*}
\begin{aligned}
%\begin{array}{ll}
dW_1^{\mathbb{Q}}(t)=dW_1(t)+\gamma_1(t)dt=dW_1(t)+(\vartheta-\sigma \rho I)\sqrt{V(t)}dt,\\
dW_2^{\mathbb{Q}}(t)=dW_2(t)+\gamma_2(t)dt=dW_2(t)+\lambda_1\sqrt{V(t)}/\sigma dt
%\end{array}
\end{aligned}
\end{eqnarray*}
and
$$dW_3^{\mathbb{Q}}(t)=dW_3(t)+\gamma_3(t)dt=dW_3(t)+\lambda_2\sqrt{V(t)}/\sigma dt.$$

Consequently, the stock price, the volatility process and the interest rate process at time $t$ under $\mathbb{Q}$ can be rewritten as follows
\begin{equation}\label{neutral}
\left\{
\begin{array}{ll}
dS(t)=r(t)S(t-)dt+\sqrt{V(t)}S(t-)dW^{\mathbb{Q}}_{1}(t)+\int_R (e^x-1)S(t-)\widetilde{N}^{\mathbb{Q}}_x(dt,dx),\quad S(0)=S_0>0,\\
dV(t)=\kappa^{\mathbb{Q}}(\theta^{\mathbb{Q}}-V(t))dt+\sigma\sqrt{V(t)}\left(\rho dW^{\mathbb{Q}}_{1}(t)+\sqrt{1-\rho^2} dW^{\mathbb{Q}}_{2}(t)\right),\quad V(0)=V_0>0,\\
dr(t)=\alpha^{\mathbb{Q}}(\beta^{\mathbb{Q}}-r(t))dt+\eta \sqrt{r(t)}dW^{\mathbb{Q}}_{3}(t),\quad r(0)=r_0>0,
\end{array}
\right.
\end{equation}
where $dW^{\mathbb{Q}}_{i}(t)=dW_{i}(t)+\gamma_{i}(t)dt$ is a Brownian motion under $\mathbb{Q}$ for $i=1,2,3$,
$$
\widetilde{N}^{\mathbb{Q}}_x(dt,dx)=N_x(dt,dx)-v_x^\mathbb{Q}(dx)dt, \quad  v_x^\mathbb{Q}(dx)dt=\int_R e^zv_{x,z}(dx\times dz)dt,
$$
and
\begin{equation*}
%\begin{array}{ll}
\kappa^{\mathbb{Q}}=\kappa+\rho\sigma(\vartheta-\sigma \rho I)+\lambda_1\sqrt{1-\rho^2},\quad \theta^{\mathbb{Q}}=\frac{\kappa \theta}{\kappa^{\mathbb{Q}}},\quad \alpha^{\mathbb{Q}}=\alpha+\lambda_{2}, \quad \beta^{\mathbb{Q}}=\frac{\alpha \beta}{\alpha+\lambda_{2}}
%\end{array}
\end{equation*}
are four risk-neutral parameters.
\begin{remark}\label{remarkeee}
Clearly,  $\tilde{N}_x$ and $\tilde{N}^{\mathbb{Q}}_x$ appeared  in (\ref{real}) and (\ref{neutral}) should satisfy the following conditions
$$
\int_R (e^x-1)v_x(dx)<+\infty,\quad \int_R (e^x-1)v_x^\mathbb{Q}(dx)<+\infty.
$$
\end{remark}

\subsection{Pricing formula for variance swaps }

We recall that a variance swap is a forward contract on the future realized variance (for short RV) of returns for a specified asset with a maturity $T>0$ and a constant strike level $K>0$. At the maturity time $T$, the payoff of a variance swap can be evaluated as
$V(T)=(RV-K)\times L,$ where $L$ is the notional amount of the swap in dollars and RV is the sum of squared returns of asset. In the risk-neutral world, the value of a variance swap is the conditional expectation of its future payoff with respect to $\mathbb{Q}$ defined by
$$V(t)=\mathbb{E}^{\mathbb{Q}}\left[e^{-\int_t^T r(s)ds }(RV-K)\times L\big| \mathcal{F}_t\right].$$
Since it is defined in the class of forward contracts, we know that $V(0)=0$. The calculation of above expectation is difficult due to it involves the joint distribution of the interest rate and the future payoff. Noticing that the price of a $T$-maturity zero-coupon bond at $t=0$ is given by
$$P(0,T)=\mathbb{E}^{\mathbb{Q}}\left[e^{-\int_0^T r(s)ds }| \mathcal{F}_0\right],$$
we can consider the pricing problem under the $T$-forward measure $\mathbb{Q}^{T}$ instead of the risk-neutral measure $\mathbb{Q}$ and so
$$V(0)=\mathbb{E}^{\mathbb{Q}}\left[e^{-\int_0^T r(s)ds }(RV-K)\times L\big| \mathcal{F}_0\right]=P(0,T)\mathbb{E}^{T}\left[(RV-K)\times L\big| \mathcal{F}_0\right],$$
where $\mathbb{E}^{T}\left[\cdot|\mathcal{F}_0\right]$ denotes the expectation under $\mathbb{Q}^{T}$ with respect to $\mathcal{F}_0$ at $t=0$.
Thus, the fair delivery strike price of a variance swap is given by
$$K_{V}=\mathbb{E}^{T}[RV| \mathcal{F}_0].$$

To calculate $K_{V}$, we study the system \eqref{neutral} under the $T$-forward measure $\mathbb{Q}^{T}$. By applying the term structure theory of interest rate (see, e.g., \cite{Brigo}), we get
$$
N_{1,t}=e^{\int_{0}^{t}r(s)ds}
$$
under the measure $\mathbb{Q}$ and
$$
N_{2,t}=P(t,T)=A(t,T)e^{-B(t,T)r(t)}
$$
under the measure $\mathbb{Q}^{T}$, where
$$A(t,T)=\left[\frac{2e^{\alpha^{Q}+\sqrt{(\alpha^{Q})^{2}+2\eta^{2}}\frac{T-t}{2}}\sqrt{(\alpha^{Q})^{2}+2\eta^{2}}}
{2\sqrt{(\alpha^{Q})^{2}+2\eta^{2}}+\left(\alpha^{Q}+\sqrt{(\alpha^{Q})^{2}+2\eta^{2}}\right)\left(e^{(T-t)\sqrt{(\alpha^{Q})^{2}+2\eta^{2}}}-1\right)}\right]^{\frac{2\alpha^{Q}\beta^{Q} }{\eta^{2}}},$$
and
$$B(t,T)=\frac{e^{(T-t)\sqrt{(\alpha^{Q})^{2}+2\eta^{2}}}-1}{2\sqrt{(\alpha^{Q})^{2}+2\eta^{2}}+\left(\alpha^{Q}+\sqrt{(\alpha^{Q})^{2}+2\eta^{2}}\right)
\left(e^{(T-t)\sqrt{(\alpha^{Q})^{2}+2\eta^{2}}}-1\right)}.$$
Using It\^{o} Lemma, we have
$$d\ln N _{1,t}=r(t)dt=\left(\int_0^t \alpha^{\mathbb{Q}}(\beta^{\mathbb{Q}}-r(s))ds+\eta \sqrt{r(s)}dW^{\mathbb{Q}}_{3}(s)\right)dt$$
and
$$d\ln N _{2,t}=\left[\frac{A'(t,T)}{A(t,T)}-B'(t,T)r(t)-B(t,T)\alpha^{\mathbb{Q}}(\beta^{\mathbb{Q}}-r(t))\right]dt-B(t,T)\eta\sqrt{r(t)}dW_3^Q(t),$$
where $A'(t,T)$ and $B'(t,T)$ are the partial derivatives with respect to $t$.

From the above discussion, we have the following result.
\begin{lemma}\label{l+}
System (\ref{neutral}) can be transformed into the following system under $\mathbb{Q}^{T}$:
\begin{equation}\label{forward}
\left\{
\begin{aligned}
dS(t)&=r(t)S(t-)dt+\sqrt{V(t)}S(t-)dW^{*}_{1}(t)+\int_R (e^x-1)S(t-)\widetilde{N}^{*}_x(dt,dx),\quad S(0)=S_0>0,\\
dV(t)&=\kappa^{*}(\theta^{*}-V(t))dt+\sigma\sqrt{V(t)}(\rho dW^{*}_{1}(t)+\sqrt{1-\rho^2} dW^{*}_{2}(t)),\quad V(0)=V_0>0,\\
dr(t)&=\left[\alpha^{*}\beta^{*}-\left(\alpha^{*}+B(t,T)\eta^{2}\right)r(t)\right]dt+\eta \sqrt{r(t)}dW^{*}_{3}(t),\quad r(0)=r_0>0,
\end{aligned}
\right.
\end{equation}
where
$$
dW^{*}_{1}(t)=dW^{\mathbb{Q}}_{1}(t),\quad \widetilde{N}^{*}_x(dt,dx)=\widetilde{N}^{\mathbb{Q}}_x(dt,dx), \quad \kappa^{*}=\kappa^{\mathbb{Q}},\quad \theta^{*}=\theta^{\mathbb{Q}}, \quad dW^{*}_{2}(t)=dW^{\mathbb{Q}}_{2}(t)$$
and
$$\quad\alpha^{*}=\alpha^{\mathbb{Q}},\quad \beta^{*}=\beta^{\mathbb{Q}},\quad dW^{*}_{3}(t)=dW^{\mathbb{Q}}_{3}(t)+B(t,T)\eta \sqrt{r(t)}dt.$$
\end{lemma}
\begin{remark}
 Lemma \ref{l+} shows that the transformation of measure only depends on the interest rate processes.
 \end{remark}

\begin{proposition}\label{31}
Let $X(t)=\ln S(t)$ be the log-price process and $B(t)=B(T-t,T)$. Then the joint moment-generating function of joint processes $X(t)$, $V(t)$ and $r(t)$ in (\ref{forward}) can be defined as follows
$$U(t,X,V,r):=\mathbb{E}^{T}\left[\exp(\omega X_T+ \varphi V_T+ \psi r_T+\chi)|X(t)=X,V(t)=V,r(t)=r\right],$$
where $\varphi$, $\psi$ and $\chi$ are constant parameters. Moreover, if
$$\frac{\sigma-\sqrt{\sigma^2+4{\kappa}^{*2}}}{2\sigma}<\omega\leq 0,\quad \varphi\leq 0, \quad \psi\leq 0,$$
then the value of $U(\tau,X,V,r)$ at $\tau := T-t$ is given by
$$U(\tau,X,V,r)=\exp(\omega X+C(\tau;q)V+D(\tau;q)r+E(\tau;q)),$$
where $q=(\omega,\varphi,\psi,\chi)$ and
$C(\tau;q),D(\tau;q),E(\tau;q)$ satisfy
\begin{eqnarray}\label{odes}
\left\{
\begin{aligned}
&\frac{d C(\tau;q)}{d \tau}=\frac{1}{2}\sigma^{2}C^{2}(\tau;q)+(\rho\sigma\omega-\kappa^*)C(\tau;q)+\frac{1}{2}(\omega^2-\omega),\\
&\frac{d D(\tau;q)}{d\tau}=\frac{1}{2}\eta D^{2}(\tau;q)-(\alpha^{*}+B(\tau)\eta^2)D(\tau;q)+\omega,\\
&\frac{d E(\tau;q)}{d\tau}=\kappa^*\theta^{*} C(\tau;q)+\alpha^{*}\beta^{*} D(\tau;q)+\int_{R}\left[(e^{\omega x}-1)-\omega(e^{x}-1)\right]\nu^{*}_{x}(dx)
\end{aligned}
\right.
\end{eqnarray}
with initial conditions
\begin{equation}\label{conditions}
C(0;q)=\varphi,\quad D(0;q)=\psi,\quad E(0;q)=\chi.
\end{equation}
\end{proposition}
\begin{proof}
It follows from (\ref{forward}) that
\begin{equation}\label{forward_X}
\left\{
\begin{aligned}
dX(t)&=\left(r(t)-\frac{1}{2}V(t)\right)dt+\sqrt{V(t)}dW^{*}_{1}(t)+\int_R (e^x-1)\widetilde{N}^{*}_x(dt,dx),\quad S(0)=S_0>0,\\
dV(t)&=\kappa^{*}(\theta^{*}-V(t))dt+\sigma\sqrt{V(t)}(\rho dW^{*}_{1}(t)+\sqrt{1-\rho^2} dW^{*}_{2}(t)),\quad V(0)=V_0>0,\\
dr(t)&=\left[\alpha^{*}\beta^{*}-\left(\alpha^{*}+B(t,T)\eta^{2}\right)r(t)\right]dt+\eta \sqrt{r(t)}dW^{*}_{3}(t),\quad r(0)=r_0>0.
\end{aligned}
\right.
\end{equation}

Next we prove that $\{U(t,X_t,V_t,r_t)\}_{0\leq t\leq T}$ is an $\mathcal{F}_t$-martingale. In fact,  by the Markov property of $X,V,r$, one has
\begin{align*}
\mathbb{E}^T[U(t,X_t,V_t,r_t)|\mathcal{F}_s]&=\mathbb{E}^T[\mathbb{E}^T[\exp(\omega X_T+ \varphi V_T+ \psi r_T+\chi)|\mathcal{F}_t]|\mathcal{F}_s]\\
&=\mathbb{E}^T[\exp(\omega X_T+ \varphi V_T+ \psi r_T+\chi)|\mathcal{F}_s]\\
&=U(s,X_s,V_s,r_s).
\end{align*}
Thus, it suffices to prove that
 $$\mathbb{E}^T[U(t,X_t,V_t,r_t)]=\mathbb{E}^{T}\left[\exp(\omega X_T+ \varphi V_T+ \psi r_T+\chi)\right]<+\infty$$
for each $0\leq t\leq T$. Denote
\begin{align*}
I_1=&\mathbb{E}^T\left[\exp\left(\varphi V_T-\frac{\omega}{2}\int_{0}^{T}V_tdt+\omega\int_{0}^{T}\sqrt{V_t}dW_1(t)\right)\right],\\
I_2=&\mathbb{E}^T\left[\exp\left(\psi r_T+\omega\int_{0}^{T}r(t)dt\right)\right],\\
I_3=&\mathbb{E}^T\left[\exp\left(\omega\int_{0}^T\int_R (e^x-1)\widetilde{N}^{*}_x(dt,dx)\right)\right].
\end{align*}
Since $\varphi\leq 0$ and $\frac{\sigma-\sqrt{\sigma^2+4{\kappa}^{*2}}}{2\sigma}<\omega\leq 0$, we know that  $\frac{\omega^2-\omega}{2}<\frac{{\kappa}^{*2}}{2\sigma^2}$ and so Corollary 5 (for the one-dimensional case) in \cite{Chr} shows that
\begin{eqnarray*}
I_1&=&\mathbb{E}^T\left[\mathbb{E}^T\left[\exp\left(\varphi V_T-\frac{\omega}{2}\int_{0}^{T}V_tdt+\omega\int_{0}^{T}\sqrt{V_t}dW_1(t)\right)\bigg|V(t)_{0\leq t\leq T}\right]\right]\\
&\leq&\mathbb{E}^T\left[\exp\left(\frac{\omega^2-\omega}{2}\int_{0}^{T}V_tdt\right)\right]< +\infty.
\end{eqnarray*}
By conditions $\omega\leq 0$ and $\psi\leq 0$,  the positive property of $r$ implies that
$I_2\leq 1$.   Since $1-e^{\omega(e^x-1)}\leq -\omega(e^x-1)$, it is seen by Remark \ref{remarkeee} and Proposition 11.2.2.5 in \cite{Mon} that
$I_3< +\infty$ and so
$$
\mathbb{E}^{T}\left[\exp(\omega X_T+ \varphi V_T+ \psi r_T+\chi)\right]=I_1I_2I_3<+\infty,
$$
which shows that $\{U(t,X_t,V_t,r_t)\}_{0\leq t\leq T}$ is an $\mathcal{F}_t$-martingale.

Now by applying It\^{o} Lemma to $U(t,X_t,V_t,r_t)$, we can obtain a partial integral-differential equation (PIDE) for $U(t,X,V,r)$ as follows
\begin{eqnarray*}
%\begin{array}{cl}
0&=&  \frac{\partial U}{\partial t}+\left(r-\frac{1}{2}V\right)\frac{\partial U}{\partial X}+[\kappa^{*}(\theta^{*}-V)]\frac{\partial U}{\partial V}+\left[\alpha^{*}\beta^{*}-(\alpha^{*}+B(t,T)\eta^{2})r\right]\frac{\partial U}{\partial r}\nonumber\\
&& \mbox{} +\frac{1}{2}V\frac{\partial^{2} U}{\partial X^{2}}
+\rho \sigma V\frac{\partial^{2} U}{\partial X \partial V}
+\frac{1}{2}\sigma^{2}V\frac{\partial^{2} U}{\partial V^{2}}
+\frac{1}{2}\eta^{2}r\frac{\partial^{2} U}{\partial r^{2}}\nonumber\\
&& \mbox{} +\int_{R}\left[U(t-,X+x,V,r)-U(t-,X,V,r)-(e^{x}-1)\frac{\partial U}{\partial X}\right]\nu^{*}_{x}(dx).
%\end{array}
\end{eqnarray*}
Denoting $\tau=T-t$ , we get%and using terminal condition $U(T,X,V,r)=exp(\omega X_T+\beta V_T+\mu )$
\begin{eqnarray}\label{PIDEt}
\frac{\partial U}{\partial \tau}&= & (r-\frac{1}{2}V)\frac{\partial U}{\partial X}+[\kappa^{*}(\theta^{*}-V)]\frac{\partial U}{\partial V}+\left[\alpha^{*}\beta^{*}-(\alpha^{*}+B(\tau)\eta^{2})r\right]\frac{\partial U}{\partial r}\nonumber\\
&& \mbox{} +\frac{1}{2}V\frac{\partial^{2} U}{\partial X^{2}}
+\rho \sigma V\frac{\partial^{2} U}{\partial X \partial V}
+\frac{1}{2}\sigma^{2}V\frac{\partial^{2} U}{\partial V^{2}}
+\frac{1}{2}\eta^{2}r\frac{\partial^{2} U}{\partial r^{2}}\nonumber\\
&& \mbox{} +\int_{R}\left[U(\tau,X+x,V,r)-U(\tau,X,V,r)-(e^{x}-1)\frac{\partial U}{\partial X}\right]\nu^{*}_{x}(dx).
\end{eqnarray}

Thanks to the affine structure in the SVSJ model, (\ref{PIDEt}) admits an analytic solution of the following form:
\begin{equation}\label{affine}
U(\tau,X,V,r)=\exp(\omega X+C(\tau;q)V+D(\tau;q)r+E(\tau;q))
\end{equation}
%Then using terminal condition $U(T,X,V,r)=exp(\omega X_T+ \varphi V_T+ \psi r_T+\chi)$, we have
with the initial condition %of equation (\ref{PIDEt}):
$$U(0,X,V,r)=\exp(\omega X+ \varphi V+ \psi r+\chi).$$
Combining (\ref{affine}) with (\ref{PIDEt}), we find that $C(\tau;q),D(\tau;q),E(\tau;q)$ satisfy  system (\ref{odes}) with initial conditions (\ref{conditions}).
\end{proof}

\begin{remark} We note that
$\frac{\sigma-\sqrt{\sigma^2+4{\kappa}^{*2}}}{2\sigma}<\omega\leq 0,\varphi\leq 0$ and $\psi\leq 0$ are sufficient (but not necessary) conditions to ensure that $\{U(t,X_t,V_t,r_t)\}_{0\leq t\leq T}$ is an $\mathcal{F}_t$-martingale. Using the terminology in \cite{Duffie}, it is easy to see that \eqref{odes},  \eqref{forward_X}, and \eqref{affine} are ``well-behaved'' at $(q,T)$.
\end{remark}

\begin{proposition}\label{32}
If all the conditions of Theorem \ref{31} holds, then $U(\tau,X,V,r)$ can be expressed as
$$U(\tau,X,V,r)=\exp(\omega X+C(\tau;q)V+D(\tau;q)r+E(\tau;q)),$$
where
\begin{align*}
&q=(\omega,\varphi,\psi,\chi),\\
&C(\tau;q)=\frac{\varphi\left(\xi_- \exp(-\zeta\tau)+\xi_+\right)+(\omega^2-\omega)\left(1-\exp(-\zeta\tau)\right)}{(\xi_+ +\varphi\sigma^2)\exp(-\zeta \tau)+\xi_- -\varphi\sigma^2},\\
&D(\tau;q)=G(\tau)F(\tau)^{-1},\\
&E(\tau;q)=\chi-\frac{2\kappa^*\theta^*}{\sigma^2}\left[\xi_+\tau+2\ln\frac{(\xi_+ +\varphi\sigma^2)\exp(-\zeta \tau)+\xi_- -\varphi\sigma^2}{2\zeta}\right]+\alpha^{*}\beta^{*} \int_0^\tau G(u)F(u)^{-1}du+J\tau,\\
&\xi_\pm=\zeta\mp (\kappa^*-\rho\sigma\omega),\\
&\zeta=\sqrt{(\kappa^*-\rho\sigma\omega)^2+\sigma^2(\omega-\omega^2)},\\
&J=\int_{R}\left[(e^{\omega x}-1)-\omega(e^{x}-1)\right]\nu^{*}_{x}(dx)
\end{align*}
with
\begin{equation*}
 \begin{bmatrix}F(\tau) \\ G(\tau) \end{bmatrix}
   =\mathcal{T} \exp\begin{bmatrix}\frac{1}{2}\int_0^\tau(\alpha^*+B(t)\eta^2)dt & -\frac{1}{2}\eta\tau \\ -\omega\tau & -\frac{1}{2}\int_0^\tau(\alpha^*+B(t)\eta^2)dt \end{bmatrix},
 % \mu^{-}=\begin{bmatrix}7 \\ 4 \end{bmatrix},
 %  \Sigma^{-}=\begin{bmatrix}100 &0 \\0 & 25\end{bmatrix}.
\end{equation*}
and $\mathcal{T }\exp$ denotes the time ordered exponential.
\end{proposition}
\begin{proof}
For simplicity, $C(\tau;q)$, $D(\tau;q)$ and $E(\tau;q)$ are replaced with $C(\tau)$, $D(\tau)$ and $E(\tau)$, respectively.  By the assumption
$C(\tau)=-\frac{2I'(\tau)}{\sigma^2 I(\tau)}$,  we know that the Riccati equation
$$
\left\{
\begin{aligned}
&\frac{d C(\tau)}{d \tau}=\frac{1}{2}\sigma^{2}C^{2}(\tau)+(\rho\sigma\omega-\kappa^*)C(\tau)+\frac{1}{2}(\omega^2-\omega),\\
&C(0)=\varphi
\end{aligned}
\right.
$$
can be transformed to the following form
$$ \left\{
\begin{aligned}
&I''(\tau)-(\rho\sigma\omega-\kappa^*)I'(\tau)+\frac{1}{4}\sigma^2(\omega^2-\omega)I(\tau)=0,\\
&-\frac{2I'(0)}{\sigma^2 I(0)}=\varphi.
\end{aligned}
\right.
$$
It follows that
$$I(\tau)=\frac{I(0)}{2 \zeta} \left[(\xi_+ + \varphi \sigma^2)\exp(-\frac{1}{2}\xi_- \tau)+(\xi_- -\varphi \sigma^2)\exp(-\frac{1}{2}\xi_+ \tau)\right],$$
where
$$ \zeta=\sqrt{(\kappa^*-\rho\sigma\omega)^2+\sigma^2(\omega-\omega^2)},\quad \xi_\pm=\zeta\mp (\kappa^*-\rho\sigma\omega).$$
Therefore,
$$C(\tau)=\frac{\varphi\left(\xi_- \exp(-\zeta\tau)+\xi_+\right)+(\omega^2-\omega)(1-\exp(-\zeta\tau))}{(\xi_+ +\varphi\sigma^2)\exp(-\zeta \tau)+\xi_- -\varphi\sigma^2}.$$

Let $K(\tau)=-D(\tau)$ for $0\leq\tau\leq T$.  Then it follows that
\begin{equation}\label{temp}
\frac{dK(\tau)}{d\tau}=-\frac{1}{2}\eta K^{2}(\tau)-(\alpha^{*}+B(\tau)\eta^2)K(\tau)-\omega
\end{equation}
with the initial condition $K(0)=-\psi$. Since $\omega\leq 0$ and $\psi\leq 0$,  by Theorem 8.5 in \cite{Gerhard}, we know that the solution of (\ref{temp}) exists for $0\leq\tau\leq T$.  Furthermore, Theorem 3.2 of \cite{Gerhard} implies that $K(\tau)$ can be expressed as follows
$$K(\tau)=\frac{G(\tau)}{F(\tau)},$$
where $G(\tau)$ and $F(\tau)$ satisfy the following differential equations:
\begin{equation}\label{h19}
  -\frac{d}{d\tau}\begin{bmatrix}F(\tau) \\ G(\tau) \end{bmatrix}
   =\begin{bmatrix}-\frac{1}{2}(\alpha^*+B(\tau)\eta^2) & \frac{1}{2}\eta \\ \omega & \frac{1}{2}(\alpha^*+B(\tau)\eta^2) \end{bmatrix}
   \begin{bmatrix}F(\tau) \\ G(\tau) \end{bmatrix}.
 % \mu^{-}=\begin{bmatrix}7 \\ 4 \end{bmatrix},
 %  \Sigma^{-}=\begin{bmatrix}100 &0 \\0 & 25\end{bmatrix}.
\end{equation}
Clearly, it follows from \eqref{h19} that
 \begin{equation*}
  \begin{bmatrix}F(\tau) \\ G(\tau) \end{bmatrix}
   =\mathcal{T} \exp\begin{bmatrix}\frac{1}{2}\int_0^\tau(\alpha^*+B(t)\eta^2)dt & -\frac{1}{2}\eta\tau \\ -\omega\tau & -\frac{1}{2}\int_0^\tau(\alpha^*+B(t)\eta^2)dt \end{bmatrix}.
 % \mu^{-}=\begin{bmatrix}7 \\ 4 \end{bmatrix},
 %  \Sigma^{-}=\begin{bmatrix}100 &0 \\0 & 25\end{bmatrix}.
\end{equation*}
 Let
$$
J=\int_{R}[(e^{\omega x}-1)-\omega(e^{x}-1)]\nu^{*}_{x}(dx).
$$
Then, it follows from \eqref{odes} that
\begin{eqnarray*}
E(\tau)&=&\chi+\kappa^*\theta^{*} \int_0^\tau C(u)du+\alpha^{*}\beta^{*} \int_0^\tau D(u)du+J\tau\\
& =&\chi-\frac{2\kappa^*\theta^*}{\sigma^2}\left[\xi_+\tau+2\ln\frac{(\xi_+ +\varphi\sigma^2)\exp(-\zeta \tau)+\xi_- -\varphi\sigma^2}{2\zeta}\right]-\alpha^{*}\beta^{*} \int_0^\tau G(u)F(u)^{-1}du+J\tau.
\end{eqnarray*}
This completes the proof. \end{proof}
\begin{remark}\label{prpr}
We would like to mention some facts as follows:
\begin{itemize}
\item[$(i)$]{Proposition \ref{32} gives a representation of the joint moment-generating function $U(\tau,X,V,r)$ of the log-return process $X$, the volatility $V$ and the interest rate $r$. The fact that the moment-generating function can be expressed by an affine form turns out to be of fundamental importance in applications of stochastic volatility models (see, for example, \cite{Duffie,Madan,Ruan,Zheng}).}
\item[$(ii)$]{It is difficult to give an exact closed form solution for Ricatti equation (\ref{h19}) and so we should employ the numerical methods to solve it.}
\item[$(iii)$]{By Theorem 8.5 in \cite{Gerhard} and the conditions $\frac{\sigma-\sqrt{\sigma^2+4{\kappa}^{*2}}}{2\sigma}<\omega\leq 0$, $\varphi\leq 0$ and $\psi\leq 0$, we know that $C(\tau;q)$ and $D(\tau;q)$ stay non-positive for $0\leq\tau\leq T$.}
\end{itemize}
\end{remark}

Now we are in the position to price variance swaps by employing the joint moment-generating function given above.
\begin{proposition}\label{vcp}
The fair strike price of variance swaps is given by
\begin{equation}\label{price}
K_V=\left.\mathrm{NA}\times\sum_{i=1}^{n}\frac{\partial^{2}}{\partial\omega^2}\exp(C(t_{i-1};q_2)V_{0}+D( t_{i-1};q_2)r_{0}+E(t_{i-2};q_2))\right|_{\omega=0^-},
\end{equation}
where $q_2=\left(0,C(\Delta t;q_1),D(\Delta t;q_1),E(\Delta t;q_1)\right)$, $\mathrm{NA}$ is the nominal amount and $\omega=0^-$ represents the left derivative at $\omega=0$.
\end{proposition}
\begin{proof}
Recall the definition of RV, we get
$$RV=\sum_{i=1}^{N}\left(\ln(S_{t_i})-\ln(S_{t_{i-1}})\right)^{2}\times NA=\sum_{i=1}^{N}\left(\ln\frac{S_{t_i}}{S_{t_{i-1}}}\right)^{2}\times \mathrm{NA}.$$
Under the $T$-forward measure, the fair delivery price can be given as follows
$$K_V=E^{T}[RV|\mathcal{F}_0]=E^{T}\left[\sum_{i=1}^{N}\left(\ln\frac{S_{t_i}}{S_{t_{i-1}}}\right)^{2}\bigg|\mathcal{F}_0\right]\times \mathrm{NA}.$$
Using the fact that $X(t)=\ln S(t)$, we get
\begin{eqnarray}\label{pricingformula}
%\begin{array}{ll}
&& E^{T}\left[\left(\ln\frac{S_{t_i}}{S_{t_{i-1}}}\right)^{2}\bigg|\mathcal{F}_0\right]\nonumber\\
&=& \left. E^{T}\left[\frac{\partial^{2}}{\partial\omega^2}\exp(\omega(X_{t_i}-X_{t_{i-1}}))\bigg|\mathcal{F}_0\right]\right|_{\omega=0^-}\nonumber\\
&=&\left. \frac{\partial^{2}}{\partial\omega^2}E^{T}\big[\exp(\omega(X_{t_i}-X_{t_{i-1}}))|\mathcal{F}_0\big]\right|_{\omega=0^-}\nonumber\\
&=&\left. \frac{\partial^{2}}{\partial\omega^2}E^{T}\big[E^{T}[\exp(\omega X_{t_i}|\mathcal{F}_{t_{i-1}}]\exp(-\omega X_{t_{i-1}}))|\mathcal{F}_0\big]\right|_{\omega=0^-}\nonumber\\
&=&\left. \frac{\partial^{2}}{\partial\omega^2}E^{T}\big[\exp(\omega X_{t_{i-1}}+C(\Delta t;q_1)V_{t_{i-1}}+D(\Delta t;q_1)r_{t_{i-1}}
 \mbox{}+E(\Delta t;q_1))\exp(-\omega X_{t_{i-1}})|\mathcal{F}_0\big]\right|_{\omega=0^-}\nonumber\\
&=&\left. \frac{\partial^{2}}{\partial\omega^2}E^{T}\big[\exp(C(\Delta t;q_1)V_{t_{i-1}}+D(\Delta t;q_1)r_{t_{i-1}}+E(\Delta t;q_1))|\mathcal{F}_0\big]\right|_{\omega=0^-}\nonumber\\
&=&\left. \frac{\partial^{2}}{\partial\omega^2}\exp(C(t_{i-1};q_2)V_{0}+D( t_{i-1};q_2)r_{0}+E(t_{i-2};q_2))\right|_{\omega=0^-},
%\end{array}
\end{eqnarray}
where $q_1=(\omega,0,0,0)$. Thus, the fair strike price of variance swaps is the multiply of the nominal amount $\mathrm{NA}$ and the Riemann-Stieltjes integral of above formula over the sampling interval $[0,T]$.
\end{proof}

\begin{remark}
It follows from Remark \ref{prpr} $(iii)$ that $C(\Delta t;q_1)\leq 0$ and $D(\Delta t;q_1)\leq 0$. Thus,  $C(t_{i-1};q_2)$, $D( t_{i-1};q_2)$ and $E(t_{i-2};q_2)$ in (\ref{price}) can be obtained by using Theorem \ref{31} and Proposition \ref{prpr}.
\end{remark}
\begin{remark}\label{moment}
For high-order moment swaps, we also have
%\begin{equation*}
%\begin{array}{ll}
$$E^{T}\left[\left(\ln\frac{S_{t_i}}{S_{t_{i-1}}}\right)^{m}\bigg|\mathcal{F}_0\right]
= \left. \frac{\partial^{m}}{\partial\omega^m}\exp(C(t_{i-1};q_2)V_{0}+D( t_{i-1};q_2)r_{0}+E(t_{i-2};q_2))\right|_{\omega=0^-},$$
%\end{array}
%\end{equation*}
where $q_1=(\omega,0,0,0)$
and
$$q_2=\left(0,C(\Delta t;q_1),D(\Delta t;q_1),E(\Delta t;q_1)\right).$$
When $m=2$, the second-moment swaps becomes so called variance swaps which are based on the realized variance and provide protection against unexpected or unfavorable change in volatility. When $m=3$, the third-moment swaps become so called skewness swaps which are based on the realized skewness and provide protection against unexpected or unfavorable change in the symmetry of the distribution. When $m=4$, the fourth-moment swaps become so called kurtosis swaps which are based on the realized kurtosis and provide protection against unexpected or unfavorable change in the tail behaviour of the distribution. It is worth to mention that the method presented in this paper can be used to solve the fair strike price of high moment risk premium to hedge relational risks, such as skewness swaps and kurtosis swaps.
\end{remark}

\section{Stochastic volatility model with jumps and CIR Model: full correlation case}

In this section, we study the problem of pricing variance swaps under the  stochastic volatility model with jumps and CIR Model under full correlation case.
\subsection{Model reformulation}
Assume that the correlations involved in model \eqref{real} can be given by
$$(dW_{1}(t),dW_{2}(t))=\rho_{12}dt=\rho_{21}dt,\quad (dW_{1}(t),dW_{3}(t))=\rho_{13}dt=\rho_{31}dt,\quad dW_{2}(t),dW_{3}(t))=\rho_{23}dt=\rho_{32}dt,$$
where $-1\leq\rho_{ij} \leq 1$ for all $i,j=1,2,3$ which are constants.
Then, the stock price, the volatility process and the interest rate process at time $t$ in the risk-neutral probability measure $\mathbb{Q}$ can be rewritten as
\begin{equation}\label{corrneutral}
\left\{
\begin{aligned}
dS(t)=&r(t)S(t-)dt+\sqrt{V(t)}S(t-)dW^{\mathbb{Q}}_{1}(t)+\int_R (e^x-1)S(t-)\widetilde{N}^{\mathbb{Q}}_x(dt,dx),\\%\quad S(0)=S_0>0,\\
dV(t)=&\kappa^{\mathbb{Q}}(\theta^{\mathbb{Q}}-V(t))dt+\sigma\sqrt{V(t)}\left(\rho_{12} dW^{\mathbb{Q}}_{1}(t)+\sqrt{1-\rho_{12}^2} dW^{\mathbb{Q}}_{2}(t)\right),\\%\quad V(0)=V_0>0,\\
dr(t)=&\alpha^{\mathbb{Q}}(\beta^{\mathbb{Q}}-r(t))dt+\eta \sqrt{r(t)}\bigg(\rho_{13}dW^{\mathbb{Q}}_{1}(t)
+\frac{\rho_{23}-\rho_{12}\rho_{13}}{\sqrt{1-\rho_{12}^2}}dW^{\mathbb{Q}}_{2}(t)\\
&+\sqrt{1-\rho_{13}^2-\left(\frac{\rho_{23}-\rho_{12}\rho_{13}}{\sqrt{1-\rho_{12}^2}}\right)^2}dW^{\mathbb{Q}}_{3}(t)
\bigg),%\quad r(0)=r_0>0,
\end{aligned}
\right.
\end{equation}
where $dW^{\mathbb{Q}}_{i}(t)=dW_{i}(t)+\gamma_{i}(t)dt$ is a Brownian motion under the risk-neutral measure for $i=1,2,3$,
$$
v_x^\mathbb{Q}(dx)dt=\int_R e^zv_{x,z}(dx\times dz)dt, \quad \widetilde{N}^{\mathbb{Q}}_x(dt,dx)=N_x(dt,dx)-v_x^\mathbb{Q}(dx)dt,
$$
 and
\begin{equation*}
%\begin{array}{ll}
\kappa^{\mathbb{Q}}=\kappa+\lambda_{1},\quad \theta^{\mathbb{Q}}=\frac{\kappa \theta}{\kappa+\lambda_{1}},\quad \alpha^{\mathbb{Q}}=\alpha+\lambda_{2}, \quad \beta^{\mathbb{Q}}=\frac{\alpha \beta}{\alpha+\lambda_{2}}
%\end{array}
\end{equation*}
are risk-neutral parameters.

Similar to the discussion in Section 2, we obtain the system (\ref{corrneutral}) under the forward measure $\mathbb{Q}^{T}$:
\begin{equation}\label{corrforward}
\left\{
\begin{aligned}
dS(t)=&\left(r(t)-\rho_{13}B(t,T)\eta\sqrt{V(t)}\sqrt{r(t)}\right)S(t-)dt+\sqrt{V(t)}S(t-)dW^{*}_{1}(t)\\
&+\int_R (e^x-1)S(t-)\widetilde{N}^{*}_x(dt,dx),\\
dV(t)=&\left(\kappa^{*}(\theta^{*}-V(t))-\rho_{23}\sigma B(t,T)\eta\sqrt{V(t)}\sqrt{r(t)}\right)dt+\sigma\sqrt{V(t)}dW^{*}_{2}(t),\\
dr(t)=&\left(\alpha^{*}\beta^{*}-\left(\alpha^{*}+B(t,T)\eta^{2}\right)r(t)\right)dt+\eta \sqrt{r(t)}dW^{*}_{3}(t),
\end{aligned}
\right.
\end{equation}
where
$$
\begin{aligned}
&dW^{*}_{1}(t)=\rho_{13}B(t,T)\eta\sqrt{r(t)}dt+dW^{\mathbb{Q}}_{1}(t),\\
&\widetilde{N}^{*}_x(dt,dx)=\widetilde{N}^{\mathbb{Q}}_x(dt,dx),\quad
\kappa^{*}=\kappa^{\mathbb{Q}},\quad \theta^{*}=\theta^{\mathbb{Q}},
\quad\alpha^{*}=\alpha^{\mathbb{Q}},\quad \beta^{*}=\beta^{\mathbb{Q}},\\
&dW^{*}_{2}(t)=\rho_{23}B(t,T)\eta\sqrt{r(t)}dt+
\left(\rho_{12} dW^{\mathbb{Q}}_{1}(t)+\sqrt{1-\rho_{12}^2} dW^{\mathbb{Q}}_{2}(t)\right),
\end{aligned}
$$
and
$$ dW^{*}_{3}(t)=\left(\rho_{13}dW^{\mathbb{Q}}_{1}(t)
+\frac{\rho_{23}-\rho_{12}\rho_{13}}{\sqrt{1-\rho_{12}^2}}dW^{\mathbb{Q}}_{2}(t)
+\sqrt{1-\rho_{13}^2-(\frac{\rho_{23}-\rho_{12}\rho_{13}}{\sqrt{1-\rho_{12}^2}})^2}dW^{\mathbb{Q}}_{3}(t)
\right)+B(t,T)\eta \sqrt{r(t)}dt.$$

\subsection{Pricing formula for variance swaps}
Using Feynman-Kac theorem, similar to the proof of Theorem \ref{31},  we can show that the joint moment generating function $U(\tau,X,V,r)$ is governed by the following PIDE:
\begin{equation}\label{pde1}
\left\{
\begin{aligned}
\frac{\partial U}{\partial \tau}
=&\frac{1}{2}V\frac{\partial^{2} U}{\partial X^{2}}
+\frac{1}{2}\sigma^{2}V\frac{\partial^{2} U}{\partial V^{2}}
+\frac{1}{2}\eta^{2}r\frac{\partial^{2} U}{\partial V^{2}}+\rho_{12}\sigma V\frac{\partial^{2} U}{\partial X \partial V}\\
&+\rho_{13}\eta\sqrt{r(T-\tau)}\sqrt{V(T-\tau)}\frac{\partial^{2} U}{\partial X \partial r}\\
&+\left[r-\rho_{13}B(\tau)\eta\sqrt{r(T-\tau)}\sqrt{v(T-\tau)}-\frac{1}{2}V\right]\frac{\partial U}{\partial X}\\
&+\left[\kappa^{*}(\theta^{*}-V)-\rho_{23}\sigma B(\tau)\eta\sqrt{r(T-\tau)}\sqrt{V(T-\tau)}\right]\frac{\partial U}{\partial V}\\
&+\left[\alpha^{*}\beta^{*}-(\alpha^{*}+B(\tau)\eta^{2})r\right]\frac{\partial U}{\partial r}
+\rho_{23}\sigma\eta\sqrt{r(T-\tau)}\sqrt{V(T-\tau)}\frac{\partial^{2} U}{\partial V \partial r}\\
&+\int_{R}\left[U(t-,X+x,V,r)-U(t-,X,V,r)-(e^{x}-1)\frac{\partial U}{\partial X}\right]\nu^{*}_{x}(dx).
\end{aligned}
\right.
\end{equation}

We note that the techniques for solving $U(\tau,X,V,r)$ in Section 2 could not be applied to handle PIDE \eqref{pde1} because it  contains the non-affine term of $\sqrt{r(T-\tau)}\sqrt{V(T-\tau)}$.  Thus, we should employ the methods introduced by Grzelak et al. \cite{Grzelak} to estimate the value of $\sqrt{r(T-\tau)}\sqrt{V(T-\tau)}$. To this end, we need the following lemmas.

\begin{lemma} (\cite{Grzelak}) \label{expc1}
For a CIR-type process $V(t)$ driven by the stochastic differential equation
$$dV(t)=\kappa(\theta-V(t))dt+\sigma\sqrt{V(t)}dW(t),$$
the expectation and variance of $\sqrt{V(t)}$ are given by
$$\mathbb{E}\left(\sqrt{V(t)}\right)=\sqrt{2c(t)}e^{-\lambda(t)/2}\sum_{k=1}^{\infty}\frac{1}{k!}(\lambda(t)/2)^k\frac{\Gamma(\frac{1+d}{2}+k)}{\Gamma(\frac{d}{2}+k)},$$
and
$$\mathbb{V}ar\left(\sqrt{V(t)}\right)=c(t)(d+\lambda(t))-2c(t)e^{-\lambda(t)}\left(\sum_{k=1}^{\infty}\frac{1}{k!}(\lambda(t)/2)^k\frac{\Gamma(\frac{1+d}{2}+k)}{\Gamma(\frac{d}{2}+k)}\right)^2,$$
where
$$c(t)=\frac{1}{4\kappa}\sigma^2(1-e^{-\kappa t}), \quad d=\frac{4\kappa\theta}{\sigma^2},\quad \lambda(t)=\frac{4\kappa e^{-\kappa t}V(0)}{\sigma^2(1-e^{-\kappa t})},$$
with $\Gamma(k)$ being the gamma function defined by
$$\Gamma(k)=\int_0^\infty t^{k-1}e^{-t}dt.$$
\end{lemma}

\begin{lemma} (\cite{Grzelak}) \label{expc2}
The expectation of $\sqrt{V(t)}$ can be approximated by following
$$\mathbb{E}\left(\sqrt{V(t)}\right)=\sqrt{c(t)(\lambda(t)-1)+c(t)d+\frac{c(t)d}{2(d+\lambda(t))}}=:\Omega_1(t),$$
where $c(t)$, $d$ and $\lambda(t)$ are given in Lemma \ref{expc1}.
\end{lemma}

In order to get a close form expression for moment generating function, we simplify $\Omega_1(t)$ in Lemma \ref{expc2}. The expectation $\mathbb{E}\left(\sqrt{V(t)}\right)$ can be further appropriated by the following form
\begin{equation}\label{omega2}
\mathbb{E}\left(\sqrt{V(t)}\right)\approx a+b e^{-ct}=:\Omega_2(t),
\end{equation}
where $a$, $b$ and $c$ are constants. The values of parameters $a$, $b$ and $c$ can be solved via an optimization problem  as follows
$$\min_{a,b,c}\parallel\Omega_1(t)-\Omega_2(t) \parallel_p,$$
where $\parallel\cdot\parallel_p$ is a $p$-norm with $p>1$.
\begin{lemma} (\cite{Grzelak})
The values of parameters $a$, $b$ and $c$ can be estimated by
$$a=\sqrt{\theta-\frac{\sigma^2}{8\kappa}},\quad b=\sqrt{V(0)-a}, \quad
c=-\ln(b^{-1}(\Omega_1(1)-a)),
$$
where $\Omega_1(t)$ is given by Lemma \ref{expc2}.
\end{lemma}

Inspired by Lemma \ref{expc2},  the expectation $\mathbb{E}\left(\sqrt{r(t)}\right)$ for a CIR-type process $r(t)$ can be given by
$$\mathbb{E}\left(\sqrt{r(t)}\right)=\sqrt{\widetilde{c}(t)(\widetilde{\lambda}(t)-1)+\widetilde{c}(t)\widetilde{d}
+\frac{\widetilde{c}(t)\widetilde{d}}{2(\widetilde{d}+\widetilde{\lambda}(t))}}=:\widetilde{\Omega}_1(t)$$
with
$$\widetilde{c}(t)=\frac{1}{4\alpha}\eta^2(1-e^{-\alpha t}), \quad \widetilde{d}=\frac{4\alpha\beta}{\eta^2},\quad \widetilde{\lambda}(t)=\frac{4\alpha e^{-\alpha t}r(0)}{\eta^2(1-e^{-\alpha t})}$$
and
\begin{equation}\label{omega22}
\mathbb{E}\left(\sqrt{r(t)}\right)\approx \widetilde{a}+\widetilde{b} e^{-\widetilde{c}t}=:\widetilde{\Omega}_2(t),
\end{equation}
where
$$\widetilde{a}=\sqrt{\beta-\frac{\eta^2}{8\alpha}},\quad \widetilde{b}=\sqrt{r(0)-\widetilde{a}}, \quad
\widetilde{c}=-\ln(\widetilde{b}^{-1}(\widetilde{\Omega}_1(1)-\widetilde{a})).
$$

Let $\widetilde{\rho}$ denote the correlation of $\sqrt{V(t)}$ and $\sqrt{r(t)}$.  Then it is obvious that the $\mathbb{E}\left(\sqrt{V(t)}\sqrt{r(t)}\right)$ can be expressed in the following form
$$\mathbb{E}\left(\sqrt{V(t)}\sqrt{r(t)}\right)=\mathbb{E}\left(\sqrt{V(t)}\right)\mathbb{E}\left(\sqrt{r(t)}\right)+\widetilde{\rho} \mathbb{V}ar\left(\sqrt{V(t)}\right)\mathbb{V}ar\left(\sqrt{r(t)}\right).$$
By computing the It\^{o} differentials of two functions $F_1 (t,V(t))=\sqrt{V(t)}$ and $F_2 (t,r(t))=\sqrt{r(t)}$, respectively, we have
$$\left(d\sqrt{V(t)},d\sqrt{r(t)}\right)=\left(\frac{1}{2}\sigma dW_2(t),\frac{1}{2}\eta dW_3(t)\right)=\frac{1}{4}\sigma\eta\rho_{23}dt,$$
which implies that $\widetilde{\rho}=\frac{1}{4}\sigma\eta\rho_{23}$. Note that
$$\mathbb{V}ar(\sqrt{V(t)})\approx \frac{\mathbb{V}ar(V(t))}{4\mathbb{E}(V(t))}, \quad \mathbb{V}ar(\sqrt{r(t)})\approx \frac{\mathbb{V}ar(r(t))}{4\mathbb{E}(r(t))}.$$
It follows from Lemmas \ref{expc1} and \ref{expc2} that
\begin{equation}\label{psi20}
\mathbb{V}ar(\sqrt{V(t)})\approx \left[c(t)-\frac{c(t)d}{2(d+\lambda(t))}\right]=:\Psi(t)
\end{equation}
and
\begin{equation}\label{psi21}
\mathbb{V}ar(\sqrt{r(t)}) \approx \left[\widetilde{c}(t)-\frac{\widetilde{c}(t)\widetilde{d}}{2(\widetilde{d}+\widetilde{\lambda}(t))}\right]=:\widetilde{\Psi}(t).
\end{equation}
This shows that
\begin{equation}\label{vr}
\mathbb{E}\left(\sqrt{V(t)}\sqrt{r(t)}\right)\approx \frac{1}{4}\sigma\eta\rho_{23}\sqrt{\Psi(t)\widetilde{\Psi}(t)}
+\Omega_2(t)\widetilde{\Omega}_2(t),
\end{equation}
where $\Psi(t)$, $\widetilde{\Psi}(t)$, $\Omega_2(t)$ and $\Omega_2(t)$ are determined by \eqref{psi20}, \eqref{psi21}, \eqref{omega2} and \eqref{omega22}, respectively.
Thus, we conclude the following proposition.
\begin{proposition}\label{p3.1}
Let $\mathcal{E}(\tau)=\mathbb{E}\left(\sqrt{V(T-\tau)}\sqrt{r(T-\tau)}\right)$. If
$$\frac{\sigma-\sqrt{\sigma^2+4{\kappa}^{*2}}}{2\sigma}<\omega\leq 0,\quad \varphi\leq 0, \quad \psi\leq 0,$$
then  the price of variance swaps under full correlation case can be given by
\begin{equation}\label{stikeprice+}
K_V=\left.\mathrm{NA}\times\sum_{i=1}^{n}\frac{\partial^{2}}{\partial\omega^2}\exp(C(t_{i-1};q_2)V_{0}+D( t_{i-1};q_2)r_{0}+E(t_{i-2};q_2))\right|_{\omega=0^-},
\end{equation}
where $q_2=\left(0,C(\Delta t;q_1),D(\Delta t;q_1),E(\Delta t;q_1)\right)$, $\mathrm{NA}$ is the nominal amount and $\omega=0^-$ represents the left derivative at $\omega=0$, $C(t_{i-1};q_2)$ and $D(t_{i-1};q_2)$ are given by Proposition \ref{32}, and
\begin{align*}
E(\tau;q) = &\chi-\frac{2\kappa^*\theta^*}{\sigma^2}\left[\xi_+\tau+2\ln\frac{(\xi_+ +\varphi\sigma^2)\exp(-\zeta \tau)+\xi_- -\varphi\sigma^2}{2\zeta}\right]+\alpha^{*}\beta^{*} \int_0^\tau G(u)F(u)^{-1}du+J\tau\\
&+\int_0^\tau \left( \rho_{13}\eta\omega\mathcal{E}(u)B(u)+\rho_{13}\eta\omega\mathcal{E}(u)D(u;q)
-\rho_{23}\sigma\eta\mathcal{E}(u)B(u)C(u;q)
+\rho_{23}\sigma\eta\mathcal{E}(u)C(u;q)D(u;q)\right)du,
\end{align*}
%\end{equation*}
with
\begin{equation*}
 \begin{bmatrix}F(\tau) \\ G(\tau) \end{bmatrix}
   =\mathcal{T} \exp\begin{bmatrix}\frac{1}{2}\int_0^\tau(\alpha^*+B(t)\eta^2)dt & -\frac{1}{2}\eta\tau \\ -\omega\tau & -\frac{1}{2}\int_0^\tau(\alpha^*+B(t)\eta^2)dt \end{bmatrix}.
 % \mu^{-}=\begin{bmatrix}7 \\ 4 \end{bmatrix},
 %  \Sigma^{-}=\begin{bmatrix}100 &0 \\0 & 25\end{bmatrix}.
\end{equation*}
\end{proposition}
\begin{proof}
 For $\frac{\sigma-\sqrt{\sigma^2+4{\kappa}^{*2}}}{2\sigma}<\omega\leq 0,\varphi\leq 0$ and $\psi\leq 0$, we know that the joint moment-generating function of the joint process $X(t)$, $V(t)$ and $r(t)$ in system (\ref{corrforward}) at $\tau := T-t$ is given by
$$U(\tau,X,V,r)=\exp(\omega X+C(\tau;q)V+D(\tau;q)r+E(\tau;q)),$$
where $q=(\omega,\varphi,\psi,\chi)$ and
$C(\tau;q),D(\tau;q),E(\tau;q)$ satisfy the following differential equations:
\begin{eqnarray*}\label{corrodes1}
\left\{
\begin{aligned}
\frac{d C(\tau;q)}{d \tau}=&\frac{1}{2}\sigma^{2}C^{2}(\tau;q)+(\rho\sigma\omega-\kappa^*)C(\tau;q)+\frac{1}{2}(\omega^2-\omega),\\
\frac{d D(\tau;q)}{d\tau}=&\frac{1}{2}\eta D^{2}(\tau;q)-(\alpha^{*}+B(\tau)\eta^2)D(\tau;q)+\omega,\\
\frac{d E(\tau;q)}{d\tau}=&\kappa^*\theta^{*} C(\tau;q)+\alpha^{*}\beta^{*} D(\tau;q)+\rho_{13}\eta\omega \mathcal{E}(\tau)B(\tau)+\rho_{13}\eta\omega\mathcal{E}(\tau) D(\tau;q)-\rho_{23}\sigma\eta\mathcal{E}(\tau)B(\tau)C(\tau;q)\\
&+\rho_{23}\sigma\eta\mathcal{E}(\tau)C(\tau;q)D(\tau;q)+\int_{R}\left[(e^{\omega x}-1)-\omega(e^{x}-1)\right]\nu^{*}_{x}(dx)
\end{aligned}
\right.
\end{eqnarray*}
with initial conditions
\begin{equation*}\label{corrconditions}
C(0;q)=\varphi,\quad D(0;q)=\psi,\quad E(0;q)=\chi.
\end{equation*}
Since that the differential equations for $C(\tau;q)$ and $D(\tau;q)$ are similar to the equation in Proposition \ref{31},  we only need to solve the equation for $E(\tau;q)$. By employing the methods used in Proposition \ref{31}, we have
\begin{align*}
E(\tau;q)=&\chi-\frac{2\kappa^*\theta^*}{\sigma^2}\left[\xi_+\tau+2\ln\frac{(\xi_+ +\varphi\sigma^2)\exp(-\zeta \tau)+\xi_- -\varphi\sigma^2}{2\zeta}\right]-\alpha^{*}\beta^{*} \int_0^\tau G(u)F(u)^{-1}du+J\tau\\
&+\int_0^\tau \left( \rho_{13}\eta\omega\mathcal{E}(u)B(u)+\rho_{13}\eta\omega\mathcal{E}(u)D(u;q)
-\rho_{23}\sigma\eta\mathcal{E}(u)B(u)C(u;q)
+\rho_{23}\sigma\eta\mathcal{E}(u)C(u;q)D(u;q)\right)du
\end{align*}
%\end{equation*}
with
\begin{equation*}
 \begin{bmatrix}F(\tau) \\ G(\tau) \end{bmatrix}
   =\mathcal{T} \exp\begin{bmatrix}\frac{1}{2}\int_0^\tau(\alpha^*+B(t)\eta^2)dt & -\frac{1}{2}\eta\tau \\ -\omega\tau & -\frac{1}{2}\int_0^\tau(\alpha^*+B(t)\eta^2)dt \end{bmatrix}.
 % \mu^{-}=\begin{bmatrix}7 \\ 4 \end{bmatrix},
 %  \Sigma^{-}=\begin{bmatrix}100 &0 \\0 & 25\end{bmatrix}.
\end{equation*}
Similar to the proof of Theorem \ref{vcp}, we know that the price of variance swaps can be given by
\begin{equation*}\label{stikeprice}
K_V=\left.\mathrm{NA}\times\sum_{i=1}^{n}\frac{\partial^{2}}{\partial\omega^2}\exp(C(t_{i-1};q_2)V_{0}+D( t_{i-1};q_2)r_{0}+E(t_{i-2};q_2))\right|_{\omega=0^-}
\end{equation*}
and so \eqref{stikeprice+} holds.
\end{proof}

\begin{remark}
From \eqref{pricingformula} and Proposition \ref{vcp},  we know that $K_V$ can be formulated by variables $V_{t_i}$ and $r_{t_i}$.  It follows from \cite{Roslan} that variables $V_{t_i}$ and $r_{t_i}$ can be appropriated by normal random variables. Moreover, letting  $Y(t_{i-1})=C(\Delta t;q_1)V_{t_{i-1}}+D(\Delta t;q_1)r_{t_{i-1}}+E(\Delta t;q_1)$, we know that the variable $Y(t_{i-1})$ is a combination of normal variables and thus derive its characteristic function.
Thus, Proposition \ref{p3.1} shows that $K_V$ can be obtained by 
$$
K_V=NA\times \left.\frac{\partial^2}{\partial \omega^2}\mathbb{E}^T[\exp \left(Y(t_{i-1})\right)|{\mathcal{F}_0}]\right|_{\omega=0^-}.
$$
\end{remark}
\section{Numerical experiments}

We have built the variance swaps pricing framework and obtained the pricing formulas under stochastic volatility with L\'evy kernel and stochastic interest model. In order to discuss the performance of the analytical formulas for variance swaps, we present our numerical results in this section. In practice, as an example, we introduce the stochastic volatility variance gama (for short SVVG) model of L\'evy process. Our presentation follows the case of SVVG model.

\subsection{Stochastic volatility variance gama model}
The stochastic volatility variance gama  model combines the variance gamma process with a stochastic volatility process. Here the VG process introduced by Madan (see, e.g., \cite{Madan}) is a representative of infinite-activity but finite-variation jump models. The variance gamma (VG) distribution was proposed for the modelling of log return on stocks. Madan, Carr and Chang generalize this approach to non-symmetric VG distributions. The L\'evy kernel of the VG process is
$$
\nu_x(dx)=\left\{
\begin{aligned}
&\frac{1}{K}\frac{e^{-G_+x}}{x},    && \quad x>0, \\
&\frac{1}{K}\frac{e^{-G_-|x|}}{|x|}, &&  \quad x\le 0
\end{aligned}
\right.
$$
with
$$
G_{\pm}=\frac{1}{\sqrt{\frac{1}{4}v_G^2 K^2+\frac{1}{2}\sigma_G^2 K}\pm \frac{1}{2}v_G K},
$$ where $v_G$ and $\sigma_G$ are drift and volatility of Brown motion of the VG process, and $K$ is the variance rate of the gamma time change. To illustrate the VG distribution, we depict its L\'evy kernels with different model parameters and we can observe the effects of these model parameters on the difference between various kernels in Fig. \ref{VG}.

\begin{figure}[H]
  \centering
  \includegraphics[height=0.5\textwidth]{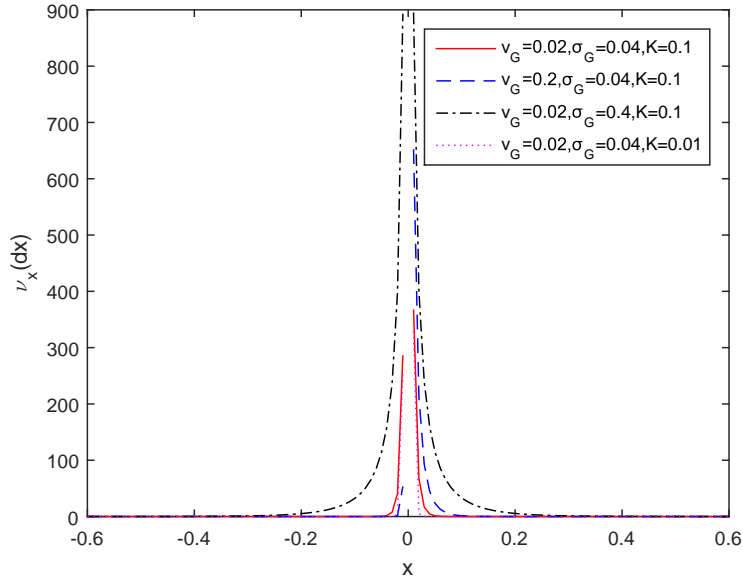}
  \caption{\label{VG}Comparison between different VG kernels with different model parameters}
  \label{referencename1}
\end{figure}

\subsection{Numerical results}

First, we show the equity premium $\phi$ of our model under the physical measure $\mathbb{P}$. In this numerical example, we set $\nu_G=0.02$, $\sigma_G=0.04$ and $K=0.01$. In addition, we take $\rho=0$ in Proposition \ref{premium} for convenience. Thus, we can obtain the form of $\phi$ as follows
\begin{equation*}\label{kerneltest}
\phi=\mu-r=\vartheta V(t)+\int_R(e^x-1)(1-e^{-\vartheta x})\nu_x(dx).
\end{equation*}
This implies that the equity premium $\phi$ is stochastic. Fig. \ref{premiumfig} describes several paths of the equity premium $\phi$. Here the parameters in $V(t)$ are taken by $V(0)=0.035$, $\kappa = 0.3$, $\theta = 0.05$ and $\sigma = 0.2$. In Fig. \ref{premiumfig}, we can find the equity premium $\phi$ is heavily dependent on the mean-reverting stochastic process $V(t)$. This result is different with \cite{Fu}, while the study of \cite{Fu} is assumed that the volatility of underlying asset is constant.
\begin{figure}[H]
  \centering
  \includegraphics[height=0.5\textwidth]{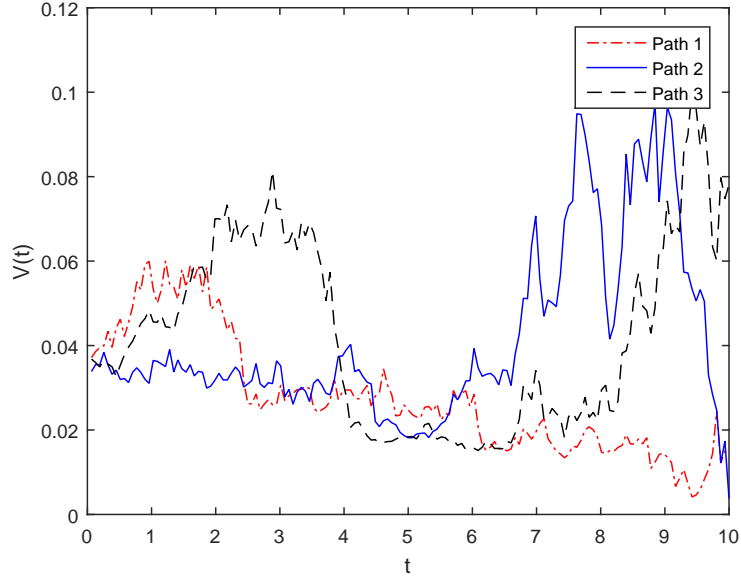}
  \caption{\label{premiumfig}Pathes of equity premium $\phi$}
  \label{referencename1}
\end{figure}
We note that the expectation of $\phi$ can be given by
\begin{equation*}\label{kerneltest}
\mathbb{E}^{\mathbb{Q}}[\phi]=\vartheta\left[e^{-\kappa t}V(0)+\theta(1-e^{-\kappa t})\right]+\int_R(e^x-1)(1-e^{-\vartheta x})\nu_x(dx).
\end{equation*}
This value does not have the stochastic term and it can be effected by the value of the risk aversion coefficient $\vartheta$. Fig. \ref{riskaverse} shows that the more risk averse $\vartheta$ is the investor, the higher is the equity premium $\phi$. This result is consistent with the one given in \cite{Ruan}, which says that the investor needs more premium when he/she is more risk-averse.

 \begin{figure}[H]
  \centering
  \includegraphics[height=0.5\textwidth]{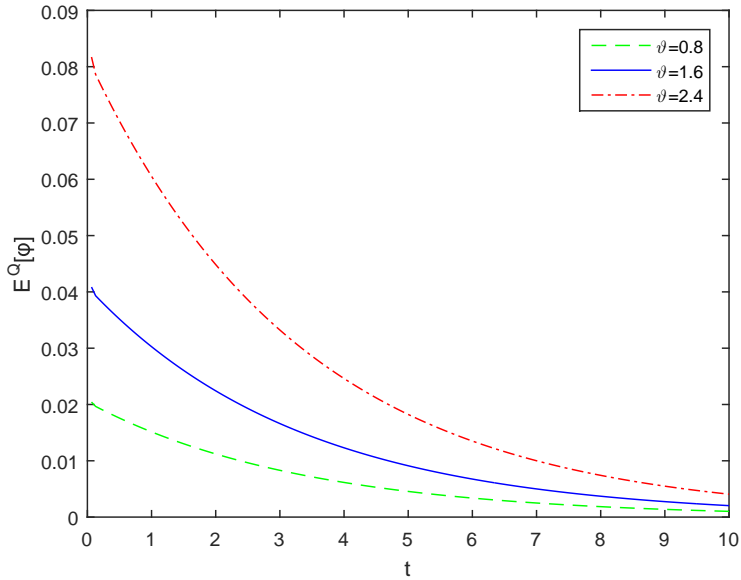}
  \caption{\label{riskaverse}The expectation of $\phi$}
  \label{referencename1}
\end{figure}
%Moreover, one can deduce that
%$$
%\int_{R}\left(e^{\omega x}-1\right)\nu_{x}(dx)=-\frac{1}{K}\log\left(1+\omega v_G \sigma_G-\frac{1}{2}\sigma_G^2K\omega^2\right).
%$$

Then, we imply Monte Carlo (for short MC) simulations to obtain numerical results as references for comparisons. The new parameters of our model under the risk-neural measure $\mathbb{Q}$ are given by Table \ref{t1}. The set of parameters was also adopted by Grzelak et al. \cite{Grzelak}. The stochastic processes of the model are discretized by using the simple Euler-Milstein scheme.
\begin{table}[H]
\centering \caption{Summary parameters of the model}\label{t1}
  \centering
  \bigskip
  \begin{tabular}{c  c  c  c  c  c  c  c  c  c  c  c  c  c}
    \hline
    % after \\: \hline or \cline{col1-col2} \cline{col3-col4} \cdots
    Parameters &   $S_{0}$ & $\rho$  &  $V_{0}$ & $\theta^{\mathbb{Q}}$   & $\kappa^\mathbb{Q}$  & $\sigma$  & $r_0$  & $\alpha^\mathbb{Q}$  & $\beta^\mathbb{Q}$ & $\eta$ &$v_G^\mathbb{Q}$ & $\sigma_G^\mathbb{Q}$ & T\\\hline
    Values   &  1 & -0.40  & $0.2236^{2}$ & $0.2236^{2}$ & 2 & 0.1 &0.05 &1.2 & 0.05 & 0.01&0.001 &0.001& 1\\
    \hline
  \end{tabular}
\end{table}

\begin{figure}[H]
  \centering
  \includegraphics[height=0.5\textwidth]{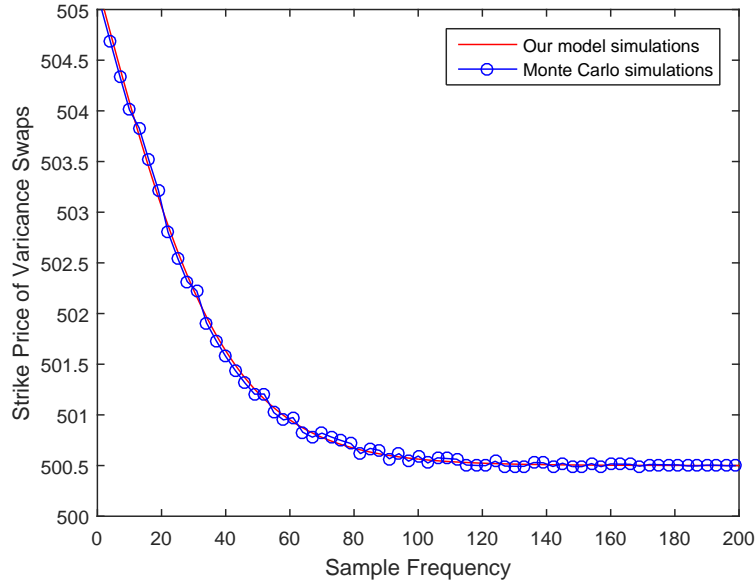}
  \caption{\label{MCMC}The comparison of our method for pricing variance swaps with MC simulations}
  \label{referencename1}
\end{figure}

Fig. \ref{MCMC} shows the comparison between the numerical results obtained from our pricing formula and
Monte Carlo simulations. Clearly, the results from our solution can match the results from the MC simulations. For example, when the number of paths reaches in MC simulations, the relative difference between the error between numerical results
obtained from our pricing formulas and MC simulations is very close already. Such a relative difference is further reduced when the number of paths increase. Also note that the convergence of the MC simulations towards our solution. This means that the result provides a verification of our solutions.

%\subsection{Numerical results}
To test the effects of the stochastic interest rate, we now calculate the fair strike values of variance swaps with stochastic
interest rate and deterministic interest rate, respectively. So we implement the analytical pricing formula with the different parameters of stochastic interest rate
to get numerical values of variance swaps. For
the variance swaps with constant deterministic interest rate, we implement the formula by Zhu and Lian (see \cite{Zhu}). We can obtain
values of variance swaps with constant deterministic interest rate by setting $\alpha^\mathbb{Q}= 0$, $ \beta^\mathbb{Q}= 0$, and $ \eta= 0$. From Fig. \ref{VSS}, we notice that with the increasing of sampling frequency, the values of variance swaps are decreasing, converging to the continuous sampling counterpart. We can also observe that, when
the spot interest rate is equal to the long-term interest rate , the values of variance swaps with
stochastic interest rate coincide with the case of constant interest rate. This implies that
the parameters $\alpha^\mathbb{Q}$ and $\eta$ have little effect on the values of variance swaps.
Finally, we can see that, when $\beta^\mathbb{Q}$ is increasing, the values of variance swaps are increasing correspondingly. The implication
is that the interest rate can impact the value of a variance swap, ignoring the effect of interest rate will
result in miss-pricing. Because interest rate is by the stochastic process, working
out the analytical pricing formula for discretely-sampled variance swaps can help pricing
variance swaps more accurately.
\begin{figure}[H]
  \centering
  \includegraphics[height=0.5\textwidth]{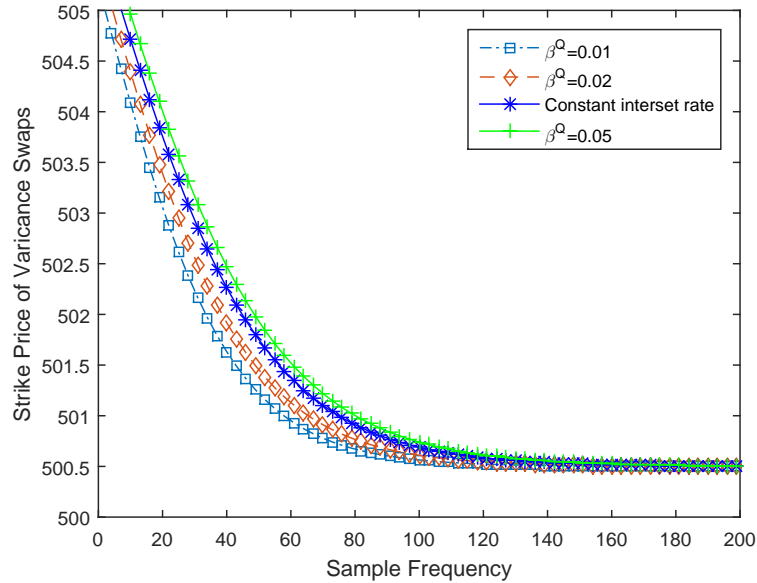}
  \caption{\label{VSS}The comparison of our method for pricing variance swaps with different $\beta^{\mathbb{Q}}$ value}
  \label{referencename2}
\end{figure}

%\begin{remark}
%The pricing formula under mentioned condition is consistent with Ruan et al. \cite{}. This result shows that the fair strike price of variance swap will not be affected by stochastic interest rate under Heston's affine assumption. Different from other transform methods \cite{}, moment methodology implies the convergence of discrete sampled to obtain the continuous strike price, and the interest rate term will not influence deriving the derivative of $\omega$ . In addition, like Remark \ref{moment}, we can solve the fair strike price of high moment risk premium to hedge relational risks, such as skewness swaps and kurtosis swaps.
%\end{remark}

\section{Conclusions}
In this paper,  we investigated the pricing of discrete variance swaps in the framework of the stochastic interest rate and stochastic volatility.  We resolved the governing PIDE and derived an analytical pricing formula based on the stochastic volatility model with jumps and the CIR model. By comparison to Monte-Carlo simulations,
our numerical results provided a verification for the correctness of the pricing formula presented in this paper. With the availability of the analytical pricing formula, we also discussed the impact of the interest rates on the values of variance swaps. We concluded that, although the variance swaps belong to volatility derivatives, it is unreasonable to ignore the impact of the stochastic interest rates. Additionally, the strike price of variance swap will increase with the rising of the parameters in the L\'evy kernel. More precisely, the existence of jump risks requires more costs to hedge and so the price of variance swaps would be higher.

We would like to point out that the pricing approach presented in this paper can be extended to some other models in connection to the stochastic interest rate and the stochastic volatility, such as the Hull-White interest rate model or the GARCH stochastic volatility model. We leave these problems for our future work.

\end{document}